%% file: cpibc-main.tex
\documentclass[a4paper,UKenglish,cleveref, autoref,thm-restate,authorcolumns]{lipics-v2019}

\usepackage[classfont=bold,funcfont=roman]{complexity} 

\usepackage{microtype,xspace,wrapfig,multicol} 
\usepackage[textsize=tiny,color=lightgray]{todonotes} 
\usepackage[normalem]{ulem} 

\usetikzlibrary{patterns,snakes}

\newcommand{\y}{\boldsymbol{y}}

\newcommand{\calT}{\mathcal{T}}
\newcommand{\calC}{\mathcal{C}}

\newcommand{\calB}{\mathcal{B}}
\newcommand{\Nset}{\mathbb{N}}
\newcommand{\Zset}{\mathbb{Z}}
\newcommand{\Hset}{\mathbb{H}}
\newcommand{\Rset}{\mathbb{R}}

\newcommand{\img}{\mathrm{Im}}

\newcommand{\gac}[2]{#1\!\!\uparrow\!\!#2}

\title{Limitations on counting in Boolean circuits and self-assembly} 

\author{Tristan Stérin}{Hamilton Institute and Department of Computer Science, Maynooth University \and \url{https://dna.hamilton.ie/tsterin/}}{tristan.sterin@mu.ie}{https://orcid.org/0000-0002-2649-3718}{Research supported by European Research Council (ERC) under the European Union’s Horizon
2020 research and innovation programme (grant agreement No 772766, Active-DNA project), and Science Foundation
Ireland (SFI) under Grant number 18/ERCS/5746.}
\author{Damien Woods}{Hamilton Institute and Department of Computer Science, Maynooth University \and \url{https://dna.hamilton.ie}}{damien.woods@mu.ie}{}{Research supported by European Research Council (ERC) under the European Union’s Horizon
2020 research and innovation programme (grant agreement No 772766, Active-DNA project), and Science Foundation
Ireland (SFI) under Grant number 18/ERCS/5746.}
\authorrunning{T. Stérin and D. Woods} 

\Copyright{Tristan Stérin and Damien Woods} 



\category{} 

\relatedversion{} 

\supplement{}

\acknowledgements{We thank Jarkko Kari for pointing us to key results on Boolean circuits and functions. We thank Christopher-Lloyd Simon for introducing us to the theory of ramification degrees and their application to quasi-bijections. We thank Constantine Evans for helpful discussions on self-assembled counters, and 
Dave Doty and Erik Winfree for discussions on IBCs over the years.}

\nolinenumbers 

\hideLIPIcs  
\begin{document}

\maketitle

\begin{abstract}
In self-assembly, a $k$-counter is a tile set that grows a horizontal ruler from left to right, containing $k$ columns each of which encodes a distinct binary string.  Counters have been fundamental objects of study in a wide range of theoretical models of tile assembly, molecular robotics and thermodynamics-based self-assembly due to their construction capabilities  using few tile types, time-efficiency of growth and combinatorial structure.  Here, we define a Boolean circuit model, called $n$-wire local railway circuits, where $n$ parallel wires are straddled by Boolean gates, each with matching fanin/fanout strictly less than $n$, and we show that such a model can not count to $2^n$ nor implement any so-called odd bijective nor quasi-bijective function.  We then define a class of self-assembly systems that includes theoretically interesting and experimentally-implemented systems that compute $n$-bit functions and count layer-by-layer.  We apply our Boolean circuit result to show that those self-assembly systems can not count to $2^n$.  This explains why the experimentally implemented iterated Boolean circuit model of tile assembly can not count to $2^n$, yet some previously studied tile system do.  Our work points the way to understanding the kinds of features required from self-assembly and Boolean circuits to implement maximal counters. 
\end{abstract}

\input{cpibc-intro.tex}
\input{cpibc-defs.tex}
\input{cpibc-bij.tex}
\input{cpibc-no-odd.tex}
\input{self_assembly}

\bibliography{cpibc-main}

\appendix
\input{app-proof_parity}

\end{document}

%% file: cpibc-intro.tex

\section{Introduction}
\begin{figure}
\includegraphics[width=\textwidth]{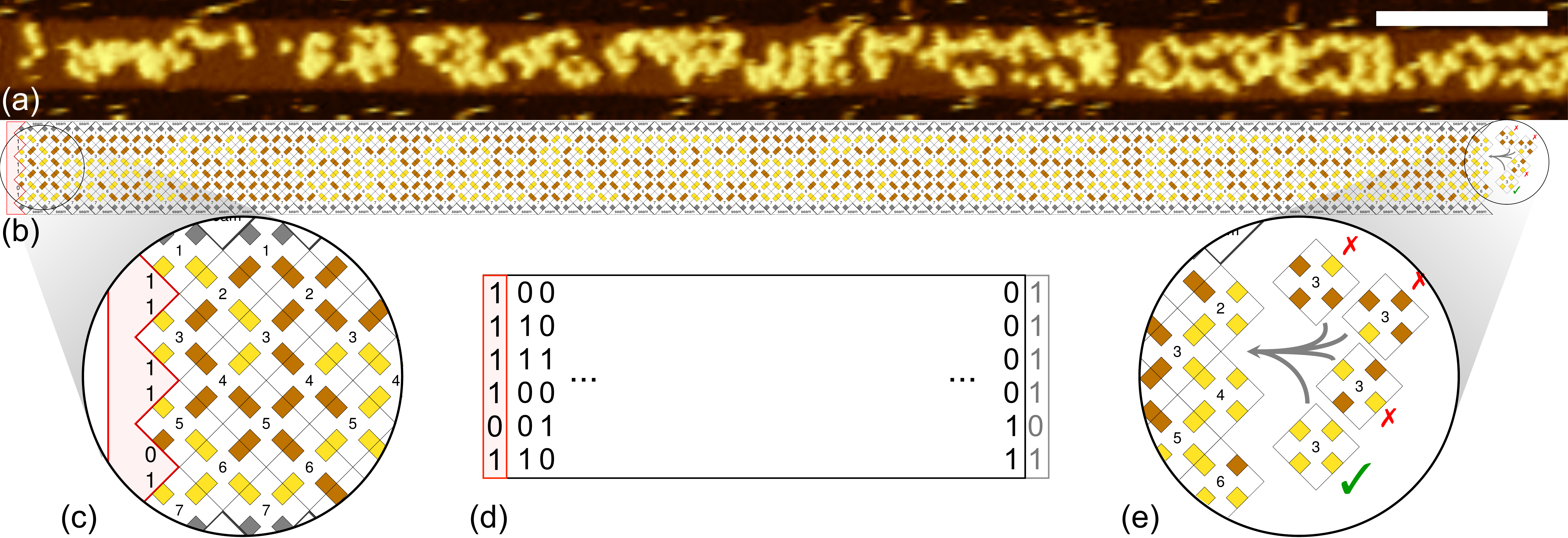}
\caption{A tile-based 63-counter in the 6-bit Iterated Boolean Circuit (IBC) model. In Section~\ref{sec:sa} we prove there is no 6-bit IBC 64-counter. 
  (a)~Atomic force microscope image of a self-assembled DNA tile 63-counter. Starting from a seed structure on the left-hand side (not shown),  DNA tiles attach to grow the assembled structure from left to right. 1-bits are labelled with yellow `circles' (streptavidin protein), 0s are unlabeled  (brown).    
Scale bar 100~nm.
  Data taken from~\cite{ibc}.
 (b)~Schematic showing intended assembled  tiled structure. 
 (c)~Zoom-in detail of (b), showing
  the first (seed) layer with the input bit sequence (from top to bottom) 111101, 
where yellow tile-glues denote 1 and brown denote 0.
 (d)~Abstract schematic of the 63-counter: 
 the 6-bit input 111101 appears on the left,
 followed by 62 distinct sequences, followed by 111101 again on the right.  
 (e)~Zoom-in of attaching tiles on right-hand side of a partially assembled structure, the attachment of a tile computes a function from two bits to two bits, and the entire tile set encodes the 63-counter algorithm. 
}
\label{fig:motivation}
\end{figure}

Both from the theoretical and the experimental points of view, counting is considered a fundamental building block for algorithmic self-assembly. 
On the theoretical side, it was established early in the field of algorithmic self-assembly that counters are a tile-efficient method to build a fixed-length ruler. Once one can make a ruler, it can be used to (efficiently) build many larger geometric shapes. For example, using an input structure containing  $O(\log n)$ square tile types, an additional (mere) constant number of tile types can then be used to first make a ruler and then an $n \times n$ square. 
Or by using a single-tile seed, a size $\Theta(\log n / \log \log n)$ tile set can go on to build $n \times n$ square in optimal expected time $\Theta(n)$~\cite{rothemund2000program,CGM04,AdChGoHu01}. These ideas easily generalise beyond squares to a wide class of more complicated geometric shapes \cite{SolWin07}. 
On some models, the combinatorial structure of tile/monomer-type efficient counters can be ``loose enough''  to allow highly-parallel construction~\cite{CGM04,nubots}, 
yet in others can be ``tight enough'' to enable large thermodynamically-stable structures~\cite{tbn}.
Counters were also used to build complex circuit patterns \cite{DBLP:conf/dna/CookRW03}, universal constructions in self-assembly \cite{USA,2HAMIU}, and as a benchmark for new self-assembly models \cite{DotPatReiSchSum10,patitz2018resiliency,fu2012self,patitz2011exact}. Hence, counters, and binary counters in particular, are fundamental to the theory of algorithmic self-assembly.\footnote{In contrast though, it should be noted that there are efficient geometry-inspired, and cellular-automata inspired,  constructions for building of shapes and or patterns that do not use counters~\cite{beckerCA}. Although they are  not as tile-type efficient as counters, they bring a more geometric, rather than counter-like information-based, flavour to shape construction.}

Experimentally, there has been a reasonable amount of effort dedicated to implementing counters~\cite{barish2005two,barish2009information,evans2014crystals,ibc}. 
An experimental piece of work \cite{ibc} (Figure~\ref{fig:motivation}), 
defined a Boolean circuit model of self-assembly, called \textit{iterated Boolean circuits} (IBCs), see Figure~\ref{fig:circuits}(a).  The model was expressive enough to permit programming of a wide range of 6-bit computations, and physical enough to permit their molecular implementation using DNA self-assembly. 
When generalised beyond 6-bit inputs to arbitrary inputs of any length $n \in \Nset$ the model is Turing universal~\cite{ibc}. 
However, despite its computational capabilities, the authors of~\cite{ibc} did not manage to find, by hand nor by computer search, any circuit that is a $2^n$ counter, or   \textit{maximal counter} meaning in the 6-bit case an iterated circuit that iterates through $2^6 = 64$ distinct bit strings before looping forever. 
Since programming requires some ingenuity, and since the search space for these circuits is huge,\footnote{There are $2^{44}$ possible 6-bit IBCs, and that number goes down to around $10^{10}$ when symmetries are taken into account.} it remained unclear whether such maximal binary counters were permitted by the model or not.
In this paper we prove they are not, and more generally give similar results on a class of Boolean circuits called \textit{railway circuits} and certain classes of self-assembly systems. 

Considering Boolean circuits, it is known since \cite{DBLP:conf/icalp/Toffoli80}, in the context of reversible computing, that adding pass-through input bits (i.e. input bits that only copy their value to the ouput) to reversible Boolean gates prevent them from implementing \textit{odd bijections}. It is essentially that result that will prevent railway circuits from implementing maximal bijective counters. However, other tools are needed to deal with \textit{quasi-bijective} maximal counters which, we show, is the only other family (besides bijections) of maximal counters.

\subsection{Results}
In Section~\ref{sec:circuits}, we define a Boolean circuit model called \textit{$n$-wire  local railway circuits}. 
An $n$-wire railway circuit consists of $n$ wires that run in straight parallel lines, with gates that straddle multiple adjacent wires (see Figure~\ref{fig:circuits}(b)) such that each gate has its fanin equal to its fanout. 
Gates are \emph{local} in the sense that no gate may straddle all $n$ wires.  
There is no restriction on the depth of these circuits.  
Railway circuits are a generalisation of IBCs and allow more possibilities for gate placement and wiring between those gates, 
yet they are restrictive enough to model a wide variety of self-assembly systems. 
Building on previous work on reversible circuits \cite{DBLP:conf/icalp/Toffoli80, DBLP:journals/corr/Xu15e, DBLP:conf/rc/BoykettKS16} and the notion of ramification degree of a function \cite{10.5555/1506267, bergeron_labelle_leroux_1997,Joyal1981UneTC}, we show that $n$-wire  local railway circuits cannot implement $2^n$ counters:
\begin{restatable}{theorem}{thmain}
\label{th:main}
For all $n > 0$, there is no local $n$-wire railway circuit that implements a $2^n$-counter.
\end{restatable}
More generally we show that no $n$-wire railway circuit implements Boolean functions  $f: \{0,1\}^n \rightarrow \{0,1\}^n$ that are  odd bijections or odd quasi-bijections (these terms are defined in Section~\ref{sec:circuits}).

We then apply these results to self-assembly in  Section~\ref{sec:sa}. 
We define a class of directed self-assembly systems that compute iterated/composed Boolean $n$-bit functions, layer-by-layer, and show that that class of self-assembly systems are simulated by railway circuits. 
Hence such systems cannot assemble maximal binary counters. 
This class includes $n$-bit IBC tile sets, hence we get: 
\begin{restatable}{theorem}{thmainibc}
\label{th:main:ibc}
For all $n \geq 3$, there is no $n$-bit IBC tile set that self-assembles a $2^n$ counter. \end{restatable}
While the layer-by-layer class of tile sets is wide enough to include an experimentally implemented IBC tile set~\cite{ibc}, and certain \textit{zigzag} systems (see Section~\ref{sec:sa}),
we also find that, from a self-assembly point of view, it is quite a restrictive class. 
Indeed building maximal counters is achievable through small, and quite reasonable, modifications to that class, 
which, in turn, highlight improvements that can be made to railway circuits to enable them to maximally count. 
Hence, this paper outlines some design principles that one \emph{should not follow} when concerned with designing maximal binary counters. 
One take home message is that in order to have an $n$-bit tile set that computes a maximal $2^n$ counter, layer-by-layer, then one should exploit some property that violates our notion of simulation by railway circuits: for example by having some tiles with fanout not equal to fanin.

\subsection{Future work}
We defined  $n$-wire local railway circuits to specifically model certain kinds of self-assembly systems.   
We leave as future work to characterise the exact family of Boolean circuits for which Theorem~\ref{th:main} holds.
That family is certainly larger than local railway circuits (for example, it would presumably include railway circuits that have gates that straddle up to $n-1$ {\em non-adjacent} wires) and goes beyond the scope of the type of circuits disscussed in \cite{DBLP:conf/icalp/Toffoli80, DBLP:journals/corr/Xu15e, DBLP:conf/rc/BoykettKS16}, but also provably does not contain the kinds of {\em railway-like} circuits that simulate the maximal $2^n$ counters from the self-assembly literature discussed in Sections~\ref{sec:zigzig} and~\ref{sec:zigzag}. 

Another direction is to find the most general class of self-assembly system for which something like Theorem~\ref{th:main:ibc} or Theorem~\ref{thm:self-assembly:no odd bijection} holds. 
Classes of self-assembly systems that are not  handled by our techniques include both undirected systems (that exploit nondeterminism in non-trivial ways to produce multiple final assemblies) and systems that do not grow in an obvious layer-by-layer fashion.
This would include systems that vary their growth pattern depending on the state of a partially grown counter structure. 
One approach is to attempt to find a more general class of circuits than railway circuits that models such general self-assembly systems.
However, it seems that a different approach might be more profitable as it is not obvious how to map such systems to a clean Boolean circuit architecture. 
  
We leave as open work to explore how our results on self-assembly generalise to higher, even or odd, alphabet sizes beyond the binary alphabet explored here.  Existing literature on reversible circuits offers pointers as they characterize the ability of local Boolean gates to implement odd/even bijections when alphabet sizes are larger than two~\cite{Boykett2016strongly}.

%% file: cpibc-defs.tex

\newcommand{\ORgate}{\textsc{Or}}
\newcommand{\ANDgate}{\textsc{And}}
\newcommand{\NOTgate}{\textsc{Not}}
\newcommand{\XORgate}{\textsc{Xor}}

\section{Counting with $n$-wire local railway circuits}\label{sec:circuits}
\input{figure-circuits.tex}

Let $n,k \in \Nset^+$. 

For $X=\{0,1,\dots,m-1\}$ and $f: X\to X$ let $\mathrm{Im}(f) = \{f(x) \, | \, x\in X\}$, 
denote the image of $f$, and let, for a finite set $Y$, $\mathrm{card}(Y)$ denote the cardinality of $Y$.

An $n$-wire, width-$k$, \emph{railway circuit} $\mathcal{C}$ is composed of $n$ parallel wires divided into $k$ sections each of width~$1$. Wires carry bits. 
A gate $g$ is specified by the tuple $(s,i,j,f_g)$ where $s \in \{ 0,1,\ldots,k-1$ is called the gate's section, 
where $0\leq i,j < n$,  
and where $f_g : \{0,1\}^{j-i+1}\to\{0,1\}^{j-i+1}$  is an arbitrary total function called the {\em gate function of $g$}.
 The gate $g=(s,i,j,f_g)$ is of width $1$, is located in section $s$, and there is exactly one gate per section.  
 The gate $g$ applies its function $f_g$ to the section's input wires between $i$ and $j$ (included). We use the notation $f_g \!\!\upharpoonright$ to refer to the extension of $f_g$ from $\{0,1\}^{j-i+1}$ to the domain $\{0,1\}^n$. The extended $f_g\!\!\upharpoonright$ simply passes through the bits on which it does not act (i.e. bits outside of the $[i,j]$ discrete interval as shown in Figure~\ref{fig:circuits}(b)). A railway circuit computes the {\em circuit function} $f_\mathcal{C}: \{0,1\}^n \to \{0,1\}^n$
by propagating its $n$ input bits from section to section and applying at each step the section's gate function to the appropriate subset of bits. In other words, we have $f_\mathcal{C} = (f_{g_{k-1}} \!\!\upharpoonright) \circ (f_{g_{k-2}} \!\!\upharpoonright) \circ \dots \circ (f_{g_1} \!\!\upharpoonright)\circ (f_{g_0} \!\!\upharpoonright)$ with $g_s$ being the gate in section $s$.
Figure~\ref{fig:circuits}(b) gives an example of a class of 6-wire railway circuits of width $7$. This example is implementing the 6-bit iterated Boolean circuit model~\cite{ibc} shown in Figure~\ref{fig:circuits}(a). 
A gate $g=(s,i,j,f_g)$ of an $n$-wire railway circuit $\mathcal{C}$ is \emph{local} if $j-i+1 < n$, i.e. the gate does not span all $n$ wires. The railway circuit $\mathcal{C}$ is \emph{local} if all of its gates are local. For instance, the railway circuit in Figure~\ref{fig:circuits}(b) is local\footnote{Note that locality does not prevent long distance influences in the circuit. If one concatenates three instances of the railway circuit in Figure~\ref{fig:circuits}(b), they obtain a new railway circuit where every input bit has an influence on every output bit: for instance, $x_0$ will influence $y_5$.}.

The following lemma defines the notion of  {\em atomic components}. Intuitively, it states that we can decompose the circuit function of a local railway circuit into a composition of functions that have properties crucial to our work.

\begin{lemma}[Atomic components]
\label{lem:pass}
Let $f_{\mathcal{C}}: \{0,1\}^n \to \{0,1\}^n$ be the circuit function of a local railway circuit $\mathcal{C}$ of width $k$. 
Then there are functions $f_0,f_1\dots,f_{k-1}$ mapping $\{0,1\}^n \to \{0,1\}^n$, called atomic components, with the following three properties:
\begin{equation}
\label{eq:one}
f_{\mathcal{C}} = f_{k-1} \circ f_{k-2} \circ \dots \circ f_0
\end{equation}
For all $0 \leq i < k$, there exists $0 \leq j < n$, such that, $\forall (x_0,\dots,x_{n-1})\in\{0,1\}^n$:
\begin{align}
\pi_j(f_i(x_0,\dots,x_{n-1})) &= x_j\label{eq:two}\\
\forall l \neq j,\, \pi_{l}(f_i(x_0,\dots,x_{j-1},0,\dots,x_{n-1})) &= \pi_{l}(f_i(x_0,\dots,x_{j-1},1,\dots,x_{n-1}))\label{eq:three}
\end{align}
where $\pi_j$ is the projection operator on the $j^{\text{th}}$ component.

\end{lemma}
\begin{proof}

Let $f_i = f_{g_i}\!\!\upharpoonright$ where $g_i$ is the gate in section $i \leq n-1$. Then, by the definition of $f_{g_i}\!\!\upharpoonright$, we have $f_\mathcal{C} = (f_{g_{k-1}} \!\!\upharpoonright) \circ\dots \circ (f_{g_1} \!\!\upharpoonright)\circ (f_{g_0} \!\!\upharpoonright) = f_{k-1} \circ f_{k-2} \circ \dots \circ f_0$ which gives Equation~\eqref{eq:one}. 

Intuitively, Equations~\eqref{eq:two} and \eqref{eq:three} state that each function $f_i$ ignores at least one of its parameters $x_j$. 
Since $\mathcal{C}$ is local, for each section $i$ the gate $g_i$ is local, meaning there is a $j$ such that wire $j$ is pass-through on section $i$ of the circuit, yielding Equations~\eqref{eq:two} and~\eqref{eq:three}.
\end{proof}

In this paper, we are interested in iterating local  railway circuits in order to count. The $i^\text{th}$ iteration of a $n$-wire railway circuit $\mathcal{C}$ is written  $f^{i}_{\mathcal{C}}(x) =\underbrace{f_\mathcal{C}(f_\mathcal{C}(\ldots f_\mathcal{C}(x))}_{i \text{ times}}$, with the convention $f^0_{\mathcal{C}}(x) = x$. Since our input space is of size $2^n$, we know that the sequence of iterations of $\mathcal{C}$ on input $x$ is periodic of period at most $2^n$. We define the \emph{trace} of $x$ (relative to $\mathcal{C}$) to be the sequence $x,f_\mathcal{C}^1(x),f_\mathcal{C}^2(x),\ldots, f_\mathcal{C}^{2^n-1}(x)$, i.e. the first $2^n$ iterations of $\mathcal{C}$ on $x$. We now define what counters are:

\newpage
\begin{definition}[$k$-counter]
\label{def:counter}
An $n$-wire railway circuit is called a \emph{$k$-counter} if it meets the following two  conditions:
\begin{enumerate} 
\item For all inputs $x\in\{0,1\}^n$, the number of distinct elements in the trace of $x$ is less or equal to $k$. 
\item There exists at least one $x\in\{0,1\}^n$ such that the number of distinct elements in the trace of input $x$ is exactly $k$. 
\end{enumerate}
\end{definition}

Since this paper is mostly concerned with proving negative results we use a relatively relaxed notion of counter that does not {\em ab initio} preclude any 2-bit string-enumerator, including counters that use the `standard' ordering on binary strings, Gray code counters, etc.
Nevertheless, we show a negative result about local railway circuits:

\thmain*

The proof of Theorem~\ref{th:main} is given in Section~\ref{sec:main}. In order to prove Theorem~\ref{th:main} we are going to describe requirements on the structure of the circuit function of a $2^n$-counter (Lemma~\ref{lem:bij}). Then, we are going to prove limitations on the ability of atomic components $f_0,\dots,f_{k-1}$ to meet those requirements (Lemma~\ref{lem:atom_restricted}). Those limitations will be stable by composition, they will transfer to the entire circuit function $f_\mathcal{C} = f_{k-1} \circ f_{k-2} \circ \dots \circ f_0$ which will end the proof.

\begin{figure}
\center
\begin{tikzpicture}[inner sep=1pt,outer sep=0pt,minimum width=1, scale=0.657]
\def \n {12}
\pgfmathsetmacro{\halfn}{0.5*\n}
\pgfmathsetmacro{\halfnPlusTwo}{\halfn+2}
\def \radius {1.3cm}
\def \margin {8}

\node[] at (-1.3,-1.2) {(a)};

\foreach \s in {1,...,\halfn} {
  \node[draw,circle,fill=blue] at ({360/\n * (\s - 1)}:\radius) {};

  \draw[->, >=stealth] ({360/\n * (\s - 1)+\margin}:\radius) arc ({360/\n * (\s - 1)+\margin}:{360/\n * (\s)-\margin}:\radius);
}
\node[draw,circle,fill=blue] at ({360/\n * (\halfnPlusTwo - 2)}:\radius) {};

\node[draw,circle,fill=black,inner sep=0pt,outer sep=0pt,minimum width=0.2] at ({360/\n * (\halfnPlusTwo - 1.65)}:\radius) {};
\node[draw,circle,fill=black,inner sep=0pt,outer sep=0pt,minimum width=0.2] at ({360/\n * (\halfnPlusTwo - 1.5)}:\radius) {};
\node[draw,circle,fill=black,inner sep=0pt,outer sep=0pt,minimum width=0.2] at ({360/\n * (\halfnPlusTwo - 1.35)}:\radius) {};

\foreach \s in {\halfnPlusTwo,...,\n} {
  \node[draw,circle,fill=blue] at ({360/\n * (\s - 1)}:\radius) {};
  \draw[->, >=stealth] ({360/\n * (\s - 1)+\margin}:\radius) 
    arc ({360/\n * (\s - 1)+\margin}:{360/\n * (\s)-\margin}:\radius);
}

\begin{scope}[shift={(5.5,0)}]
\node[] at (-1.3,-1.2) {(b)};
\node[] at (1.55,0) {$\boldsymbol{y}$};

\foreach \s in {1,...,\halfn} {
  \node[draw,circle,fill=blue] at ({360/\n * (\s - 1)}:\radius) {};

  \draw[->, >=stealth] ({360/\n * (\s - 1)+\margin}:\radius) arc ({360/\n * (\s - 1)+\margin}:{360/\n * (\s)-\margin}:\radius);
}
\node[draw,circle,fill=blue] at ({360/\n * (\halfnPlusTwo - 2)}:\radius) {};

\node[draw,circle,fill=black,inner sep=0pt,outer sep=0pt,minimum width=0.2] at ({360/\n * (\halfnPlusTwo - 1.65)}:\radius) {};
\node[draw,circle,fill=black,inner sep=0pt,outer sep=0pt,minimum width=0.2] at ({360/\n * (\halfnPlusTwo - 1.5)}:\radius) {};
\node[draw,circle,fill=black,inner sep=0pt,outer sep=0pt,minimum width=0.2] at ({360/\n * (\halfnPlusTwo - 1.35)}:\radius) {};

\foreach \s in {\halfnPlusTwo,...,\n} {
  \node[draw,circle,fill=blue] at ({360/\n * (\s - 1)}:\radius) {};
}
\pgfmathsetmacro{\tmp}{\n-1}
\foreach \s in {\halfnPlusTwo,...,\tmp } {
  \draw[->, >=stealth] ({360/\n * (\s - 1)+\margin}:\radius) 
    arc ({360/\n * (\s - 1)+\margin}:{360/\n * (\s)-\margin}:\radius);
}
\draw[->, >=stealth] (0.9,-0.57) -- (-0.96,0.58) {};
\end{scope}
\end{tikzpicture}\vspace{-1.5ex}
\caption{Each node represents a distinct $n$-bit string, each arrow represents application of a circuit function. 
The figure captures the intuition that there are only two kinds of $2^n$-counter: (a)~a cycle that repeats all bit strings forever,
 and (b)~an almost-cycle, that (if we begin at $\y$) hits all strings once, and then cycles on a smaller loop. Note that $\y$ has no antecedent.
 (a) is a bijection, (b) is a quasi-bijection.}\label{fig:all_cycles}
\end{figure}

\begin{remark}
In the following, when we talk about a function, in general we will set its domain to be $\{0,1,\dots,m-1\}$ for some arbitrary $m\neq 0$. When we consider a circuit's function, the domain of the function  is the set of strings $\{0,1\}^n$ which we will sometimes (for convenience) identify with the set of numbers $\{0,1,\dots,2^n - 1\}$, i.e. $m= 2^n$.
\end{remark}

\begin{definition}[Quasi-bijection]
\label{def:qb}
A quasi-bijection $f:\{0,1,\dots,m-1\} \to \{0,1,\dots,m-1\}$ is such that there exists exactly one $y\in\{0,1,2,\dots,m-1\}$ reached by no antecedent: $\forall x\in\{0,1,2,\dots,m-1\},\, f(x) \neq y$.

\end{definition}

\begin{remark}
By the pigeonhole argument, because there is exactly one $y$ with no antecedent in a quasi-bijection $f$, there is also exactly one $z$ which is reached by exactly two antecedents. 
\end{remark}

\begin{restatable}{lemma}{lembij}
    \label{lem:bij}
    
    The circuit function of a $2^n$-counter on $\{0,1\}^n$ is either a bijection or a quasi-bijection.

\end{restatable}

\begin{proof}
Figure~\ref{fig:all_cycles} illustrates the only two behaviors that match the definition of a $2^n$-counter (Definition~\ref{def:counter}). The case of Figure~\ref{fig:all_cycles}(a) corresponds to the circuit function being a bijection: every $x\in \{0,1\}^n$ has exactly one antecedent. The case of Figure~\ref{fig:all_cycles}(b) corresponds to the circuit function being a quasi-bijection: there is only one $y\in \{0,1\}^n$ that has no antecedent.
\end{proof}

%% file: figure-circuits.tex
\begin{figure}[t]
\begin{subfigure}{0.4\textwidth}
  \begin{tikzpicture}[scale=0.7]
\draw (0,0) -- (1,0);
\draw (1,-0.5) rectangle (2,0.5);
\node [left] at (0,0) {$x_0$};
\node at (1.5,0) {$g_0$};
\draw (2,0) -- (3,-0.5);
 \begin{scope}[shift={(0,-1.5)}]
\draw (0,0.25) -- (1,0.25);
\draw (0,-0.25) -- (1,-0.25);
\draw (1,-0.5) rectangle (2,0.5);
\node [left] at (0,0.25) {$x_1$};
\node [left] at (0,-0.25) {$x_2$};
\node at (1.5,0) {$g_1$};
\draw (2,0.25) -- (3,0.5);
\draw (2,-0.25) -- (3,-0.5);
\end{scope}
 \begin{scope}[shift={(0,-3)}]
\draw (0,0.25) -- (1,0.25);
\draw (0,-0.25) -- (1,-0.25);
\draw (1,-0.5) rectangle (2,0.5);
\node [left] at (0,0.25) {$x_3$};
\node [left] at (0,-0.25) {$x_4$};
\node at (1.5,0) {$g_2$};
\draw (2,0.25) -- (3,0.5);
\draw (2,-0.25) -- (3,-0.5);
\end{scope}
 \begin{scope}[shift={(0,-4.5)}]
\draw (0,0) -- (1,0);
\draw (1,-0.5) rectangle (2,0.5);
\node [left] at (0,0) {$x_5$};
\node at (1.5,0) {$g_3$};
\draw (2,0) -- (3,0.5);
\end{scope}
\begin{scope}[shift={(2,-3.75)}]
\begin{scope}[shift={(0,3)}]
\draw (1,-0.5) rectangle (2,0.5);
\node at (1.5,0) {$g_4$};
\draw (2,0.25) -- (3,0.25);
\draw (2,-0.25) -- (3,-0.25);
\node [right] at (3,0.25) {$y_0$};
\node [right] at (3,-0.25) {$y_1$};
\end{scope}
 \begin{scope}[shift={(0,1.5)}]
\draw (1,-0.5) rectangle (2,0.5);
\node at (1.5,0) {$g_5$};
\draw (2,0.25) -- (3,0.25);
\draw (2,-0.25) -- (3,-0.25);
\node [right] at (3,0.25) {$y_2$};
\node [right] at (3,-0.25) {$y_3$};
\end{scope}
 \begin{scope}[shift={(0,0)}]
\draw (1,-0.5) rectangle (2,0.5);
\node at (1.5,0) {$g_6$};
\draw (2,0.25) -- (3,0.25);
\draw (2,-0.25) -- (3,-0.25);
\node [right] at (3,0.25) {$y_4$};
\node [right] at (3,-0.25) {$y_5$};
\end{scope}
\end{scope}
\node at (-0.7,-7.3) {(a)};
\end{tikzpicture}
\end{subfigure}
\begin{subfigure}{0.48\textwidth}
  \begin{tikzpicture}[scale=0.48]
\draw (0,0) -- (1,0);
\draw (1,-0.5) rectangle (2,0.5);
\node [left] at (0,0) {$x_0$};
\draw (2,0) -- (8,0);
\draw (10,0) -- (15,0);
\node [left] at (0,-1.5) {$x_1$};
\draw (0,-1.5) -- (3,-1.5);
\draw (4,-1.5) -- (9,-1.5);
\draw (10,-1.5) -- (15,-1.5);
\node [left] at (0,-3) {$x_2$};
\draw (0,-3) -- (3,-3);
\draw (5,-3) -- (10,-3);
\draw (12,-3) -- (15,-3);
\node [left] at (0,-4.5) {$x_3$};
\draw (0,-4.5) -- (4,-4.5);
\draw (6,-4.5) -- (11,-4.5);
\draw (12,-4.5) -- (15,-4.5);
\node [left] at (0,-6) {$x_4$};
\draw (0,-6) -- (5,-6);
\draw (6,-6) -- (13,-6);
\node [left] at (0,-7.5) {$x_5$};
\draw (0,-7.5) -- (6,-7.5);
\draw (8,-7.5) -- (13,-7.5);
\node [left] at (16,0) {$y_0$};
\node [left] at (16,-1.5) {$y_1$};
\node [left] at (16,-3) {$y_2$};
\node [left] at (16,-4.5) {$y_3$};
\node [left] at (16,-6) {$y_4$};
\node [left] at (16,-7.5) {$y_5$};
\node at (1.5,0) {$g_0$};

\draw [
    thick,
    decoration={
        brace,
        mirror,
        raise=0.5cm
    },
    decorate
] (0,-9) -- (15,-9); 

\node [below] at (7.5, -10.3) {$f_{\mathcal{C}}$};

\draw [
    thick,
    decoration={
        brace,
        mirror,
        raise=0.5cm
    },
    decorate
] (0.7,-7.5) -- (2.3,-7.5); 

\node [below] at (1.5, -8.5) {$f_{g_0}\!\!\upharpoonright$};

\begin{scope}[shift={(2,-0.5)}]
\draw [
    thick,
    decoration={
        brace,
        mirror,
        raise=0.5cm
    },
    decorate
] (0.7,-7) -- (2.3,-7); 

\node [below] at (1.5, -8) {$f_{g_1}\!\!\upharpoonright$};
\begin{scope}[shift={(2,0)}]
\draw [
    thick,
    decoration={
        brace,
        mirror,
        raise=0.5cm
    },
    decorate
] (0.7,-7) -- (2.3,-7); 

\node [below] at (1.5, -8) {$f_{g_2}\!\!\upharpoonright$};
\begin{scope}[shift={(2,0)}]
\draw [
    thick,
    decoration={
        brace,
        mirror,
        raise=0.5cm
    },
    decorate
] (0.7,-7) -- (2.3,-7); 

\node [below] at (1.5, -8) {$f_{g_3}\!\!\upharpoonright$};
\begin{scope}[shift={(2,0)}]
\draw [
    thick,
    decoration={
        brace,
        mirror,
        raise=0.5cm
    },
    decorate
] (0.7,-7) -- (2.3,-7); 

\node [below] at (1.5, -8) {$f_{g_4}\!\!\upharpoonright$};
\begin{scope}[shift={(2,0)}]
\draw [
    thick,
    decoration={
        brace,
        mirror,
        raise=0.5cm
    },
    decorate
] (0.7,-7) -- (2.3,-7); 

\node [below] at (1.5, -8) {$f_{g_5}\!\!\upharpoonright$};
\begin{scope}[shift={(2,0)}]
\draw [
    thick,
    decoration={
        brace,
        mirror,
        raise=0.5cm
    },
    decorate
] (0.7,-7) -- (2.3,-7); 

\node [below] at (1.5, -8) {$f_{g_6}\!\!\upharpoonright$};
\end{scope}
\end{scope}
\end{scope}
\end{scope}
\end{scope}
\end{scope}

\filldraw (0.5,0) circle (3pt);
\filldraw (0.5,-3) circle (3pt);
\filldraw (0.5,-4.5) circle (3pt);
\filldraw (0.5,-6) circle (3pt);
\filldraw (0.5,-7.5) circle (3pt);
\filldraw (0.5,-1.5) circle (3pt);

\filldraw (2.5,0) circle (3pt);
\filldraw (2.5,-3) circle (3pt);
\filldraw (2.5,-4.5) circle (3pt);
\filldraw (2.5,-6) circle (3pt);
\filldraw (2.5,-7.5) circle (3pt);
\filldraw (2.5,-1.5) circle (3pt);

\begin{scope}[shift={(2,0)}]
\filldraw (2.5,0) circle (3pt);
\filldraw (2.5,-3) circle (3pt);
\filldraw (2.5,-4.5) circle (3pt);
\filldraw (2.5,-6) circle (3pt);
\filldraw (2.5,-7.5) circle (3pt);
\filldraw (2.5,-1.5) circle (3pt);
\begin{scope}[shift={(2,0)}]
\filldraw (2.5,0) circle (3pt);
\filldraw (2.5,-3) circle (3pt);
\filldraw (2.5,-4.5) circle (3pt);
\filldraw (2.5,-6) circle (3pt);
\filldraw (2.5,-7.5) circle (3pt);
\filldraw (2.5,-1.5) circle (3pt);
\begin{scope}[shift={(2,0)}]
\filldraw (2.5,0) circle (3pt);
\filldraw (2.5,-3) circle (3pt);
\filldraw (2.5,-4.5) circle (3pt);
\filldraw (2.5,-6) circle (3pt);
\filldraw (2.5,-7.5) circle (3pt);
\filldraw (2.5,-1.5) circle (3pt);
\begin{scope}[shift={(2,0)}]
\filldraw (2.5,0) circle (3pt);
\filldraw (2.5,-3) circle (3pt);
\filldraw (2.5,-4.5) circle (3pt);
\filldraw (2.5,-6) circle (3pt);
\filldraw (2.5,-7.5) circle (3pt);
\filldraw (2.5,-1.5) circle (3pt);
\begin{scope}[shift={(2,0)}]
\filldraw (2.5,0) circle (3pt);
\filldraw (2.5,-3) circle (3pt);
\filldraw (2.5,-4.5) circle (3pt);
\filldraw (2.5,-6) circle (3pt);
\filldraw (2.5,-7.5) circle (3pt);
\filldraw (2.5,-1.5) circle (3pt);
\begin{scope}[shift={(2,0)}]
\filldraw (2.5,0) circle (3pt);
\filldraw (2.5,-3) circle (3pt);
\filldraw (2.5,-4.5) circle (3pt);
\filldraw (2.5,-6) circle (3pt);
\filldraw (2.5,-7.5) circle (3pt);
\filldraw (2.5,-1.5) circle (3pt);
\end{scope}
\end{scope}
\end{scope}
\end{scope}
\end{scope}
\end{scope}

\begin{scope}[shift={(2,0)}]

\draw (2,0) -- (3,0);
\draw (0,-1.5) -- (1,-1.5);
\draw (1,-1) rectangle (2,-3.5);

\node at (1.5,-2.25) {$g_1$};
\draw (2,-1.5) -- (3,-1.5);

\begin{scope}[shift={(2,0)}]
\begin{scope}[shift={(0,-1.5)}]
\draw (0,-1.5) -- (1,-1.5);

\draw (2,-1.5) -- (3,-1.5);
\begin{scope}[shift={(0,-1.5)}]
\draw (0,-1.5) -- (1,-1.5);
\draw (1,-1) rectangle (2,-3.5);

\node at (1.5,-2.25) {$g_2$};
\draw (2,-1.5) -- (3,-1.5);
\begin{scope}[shift={(2,-1.5)}]
\draw (0,-1.5) -- (1,-1.5);

\draw (2,-1.5) -- (3,-1.5);
\begin{scope}[shift={(0,-1.5)}]
\draw (0,-1.5) -- (1,-1.5);
\draw (1,-2) rectangle (2,-1);

\node at (1.5,-1.5) {$g_3$};
\draw (2,-1.5) -- (3,-1.5);
\end{scope}
\end{scope}
\end{scope}
\end{scope}
\begin{scope}[shift={(4,0)}]
\begin{scope}[shift={(0,1.5)}]
\draw (0,-1.5) -- (1,-1.5);
\draw (1,-1) rectangle (2,-3.5);
\node at (1.5,-2.25) {$g_4$};
\draw (2,-1.5) -- (3,-1.5);
\draw (2,-3) -- (3,-3);
\end{scope}
\begin{scope}[shift={(2,-1.5)}]
\draw (0,-1.5) -- (1,-1.5);
\draw (1,-1) rectangle (2,-3.5);
\node at (1.5,-2.25) {$g_5$};
\draw (2,-1.5) -- (3,-1.5);
\draw (2,-3) -- (3,-3);
\begin{scope}[shift={(2,-3)}]
\draw (0,-1.5) -- (1,-1.5);
\draw (1,-1) rectangle (2,-3.5);
\node at (1.5,-2.25) {$g_6$};
\draw (2,-1.5) -- (3,-1.5);
\draw (2,-3) -- (3,-3);
\end{scope}
\end{scope}
\end{scope}
\end{scope}
\end{scope}
\node at (-0.7,-11) {(b)};
\end{tikzpicture}
\end{subfigure}
\caption{(a)~Iterated Boolean circuit (IBC) layer with inputs ($x_0,\ldots,x_5$), Boolean gates ($g_0,\ldots,g_6$) and  outputs ($y_0,\ldots,y_5$). The circuit computes on a 6-bit circuit input by iterating (repeating) this layer over and over. 
(b)~A railway circuit that simulates a 6-bit IBC.  The railway circuit uses $6$ wires, is of (horizontal) width $7$ and the decomposition of $f_\mathcal{C}$ into $7$ atomic components is shown with each component denoted $f_{g_0}\!\!\upharpoonright,f_{g_1}\!\!\upharpoonright$, etc. Black dots delimit sections. The railway circuit is local because none of its gates span all 6 wires.}
\label{fig:circuits}
\end{figure}

%% file: cpibc-bij.tex

\section{Ramification degrees and theory of bijective functions}
\label{sec:thf}
In order to prove limitations on the expressiveness of atomic components (Lemma~\ref{lem:atom_restricted}) we will make use of the general theory of functions and bijective functions.

\subsection{Ramification degree of a function}

We make use of, in a self-contained manner, the notion of \emph{ramification degree} of a function which has been developed much further in the field of Analytic Combinatorics \cite{10.5555/1506267, bergeron_labelle_leroux_1997,Joyal1981UneTC}.

\begin{definition}[Ramification degree]
\label{def:ram}
Take any function $f:\{0,\dots,m-1\} \to \{0,\dots,m-1\}$.
For $i\in\{0,\dots,m-1\}$, define $a_i(f)$ to be the number of 
antecedents of $i$ under $f$: $a_i(f) = \mathrm{card}\!\left(\{ j \, | \, f(j) = i \}\right )$.
Define $r_i(f)$, the ramification degree of input $i$ under $f$, to be:
$ r_i(f) = \max \left ( 0, a_i(f) - 1 \right )$
Finally, define $r(f) = \sum_{i\in \{0,\dots,m-1\}} r_i(f)$ to be the ramification degree  of the function $f$.
\end{definition}

\input{figure-potato}

We have an elegant way to describe what $r(f)$ is counting:

\begin{lemma}
\label{lem:rr}
Let $X = \{0,\dots,m-1\}$ and $f:X \to X$ then
$$ r(f) = \emph{card}(X) - \emph{card} ( \emph{Im}(f) ) = m - \emph{card}(\emph{Im}(f))$$
\end{lemma}

\begin{proof}

We are going to show that $r(f) + \text{card}(\text{Im}(f)) = m$. Figure~\ref{fig:potato} gives a general example of the situation. For $i\in X$, consider the set $f^{-1}(i)$ of the antecedents of $i$ by $f$. By definition of $f^{-1}(i)$ we have $\sum_{i\in X} \text{card}(f^{-1}(i)) = \text{card}(X)$. Now, define $J$, the set of $i$ such that $f^{-1}(i) \neq \emptyset$. By definition of $r_i(f)$, we have $r_i(f) + 1 = \text{card}(f^{-1}(i))$ when $i\in J$ and $r_i(f) = 0$ otherwise.  By definition of $\text{Im}(f)$, we have $\text{card}(\text{Im}(f)) = \text{card}(J)$. Now we have
\begin{align*} 
r(f) + \text{card}(\text{Im}(f)) &= \sum_{i\in X}r_i(f) + \text{card}(J) = \sum_{i\in J}r_i(f) + \underbrace{\sum_{i \not \in J}r_i(f)}_{0} + \text{card}(J)\\
& = \sum_{i\in J}(\text{card}(f^{-1}(i)) - 1) + \text{card}(J) = \sum_{i\in J}\text{card}(f^{-1}(i)) - \text{card}(J) + \text{card}(J)\\
& = \sum_{i\in J}\text{card}(f^{-1}(i)) = \sum_{i\in X}\text{card}(f^{-1}(i)) = \text{card}(X) = m \qedhere
\end{align*}
\end{proof}

We can easily describe functions with ramification degree $0$ and $1$:

\begin{lemma}

Let $f:\{0,\dots,m-1\} \to \{0,\dots,m-1\}$ then we have the two following equivalences:
\begin{enumerate}
\item $r(f) = 0 \Leftrightarrow f$ is a bijection.
\item $r(f) = 1 \Leftrightarrow f$ is a quasi-bijection.
\end{enumerate}
\end{lemma}
\begin{proof}
Let $X = \{0,\dots,m-1\}$.
\begin{enumerate}
\item If $r(f) = 0$, by Lemma~\ref{lem:rr} we have $\text{card}(\text{Im}(f)) = \text{card}(X)$. It means that $f$ is surjective, but $f$ has the same domain and range so $f$ is bijective.
\item If $r(f) = 1$, by Lemma~\ref{lem:rr} we have $\text{card}(\text{Im}(f)) = \text{card}(X)-1$. It means that there is exactly one $x\in X$ which is not reached by $f$ so $f$ is a quasi-bijection (see Definition~\ref{def:qb}).\qedhere
\end{enumerate}
\end{proof}

\noindent An important property of ramification degree is that it does not decrease under composition:

\begin{lemma}
\label{lem:comp}
Let $f,g \in \{0,\dots,m-1\} \to \{0,\dots,m-1\}$. Then we have:
$$ r(f\circ g) \geq \max(r(f),r(g)) $$
\end{lemma}
\begin{proof}
Let $X = \{0,\dots,m-1\}$. By Lemma~\ref{lem:rr}, we wish to show that $\text{card}( \text{Im}(f\circ g) ) \leq \text{card}(\text{Im}(f))$ and $\text{card}( \text{Im}(f\circ g) ) \leq \text{card}(\text{Im}(g))$. Firstly, we have: $\text{Im}(f\circ g) \subset \text{Im}(f)$. Hence, $\text{card}( \text{Im}(f\circ g) ) \leq \text{card}(\text{Im}(f))$. Secondly, we have $\text{Im}(f\circ g) = \{ f(x) \, | \, x \in \text{Im}(g) \}$. It follows that $\text{card}( \text{Im}(f\circ g) ) \leq \text{card}(\text{Im}(g))$.
\end{proof}

From Lemma~\ref{lem:comp}, we immediately get the following:
\begin{corollary}
\label{cor:ram}
Let $f:\{0,\dots,m-1\} \to \{0,\dots,m-1\}$ such that there exists $f_0,f_1,\dots,f_{k-1}$ with $f = f_{k-1} \circ f_{k-2} \circ \dots \circ f_0$. Then:
\begin{enumerate}
\item $r(f) = 0 \Rightarrow \forall i,\, r(f_i) = 0$
\item $r(f) = 1 \Rightarrow \forall i,\, r(f_i) = 0 \text{  or  } r(f_i) = 1$
\end{enumerate}
\end{corollary}
\begin{remark}
Said otherwise, you can only construct a bijection by composing bijections and you can only construct a quasi-bijection by composing bijections and quasi-bijections.
\end{remark}

\subsection{Bijective functions}

The following results about bijections are well-known group theoretic results which the reader can find, for instance, in \cite{rotman2012introduction}.
Here, we define a few notions that are required to state and prove out main results (Lemma~\ref{lem:atom_restricted} and Theorem~\ref{th:main}), with some details left to  
Appendix~\ref{app:proof_parity}. 

\begin{definition}[The symmetric group $\mathfrak{S}_m$]
The set of all bijections with domain and image$\{0,1,\dots,m-1\}$, is called $\mathfrak{S}_m$, the symmetric group of order $m$. It is a group for function composition $\circ$ and its neutral element is the identity.
\end{definition}

\begin{remark}
Note that the set of bijections on $\{0,1\}^n \to \{0,1\}^n$ corresponds to $\mathfrak{S}_{2^n}$.
\end{remark}

\begin{definition}[A swap]
A swap (or transposition) is a bijection $\tau \in \mathfrak{S}_m$ which leaves all its inputs invariant except for two that it swaps: i.e. there exists $i_0 \neq i_1\in \{0,1,\dots,m-1\}$ such that $\tau(i_0) = i_1$, $\tau(i_1) = i_0$ and $\tau(i) = i$ for all $i\not \in \{i_0,i_1\}$.
\end{definition}

\begin{remark}
A swap is its own inverse: $\tau \circ \tau = \text{Id}$.
\end{remark}

\begin{lemma}[Decomposition into swaps]
\label{lem:swde}
Take any $f\in \mathfrak{S}_m$. There exists $p$ swaps $\tau_0,\tau_1,\dots,\tau_{p-1}$ such that: $f = \tau_{p-1}\circ \dots \circ \tau_0$. We call $(\tau_0,\tau_1,\dots,\tau_{p-1})$ a swap-decomposition of $f$.
\end{lemma}
\begin{proof}
Another way to read $f = \tau_{p-1}\circ \dots \circ \tau_0$ is $\tau_0^{-1} \circ \tau_1^{-1} \dots \circ \tau_{p-1}^{-1} \circ f = \text{Id}$ which means that the composition of transpositions $\tau_0^{-1} \circ \tau_1^{-1} \dots \circ \tau_{p-1}^{-1} = \tau_0 \circ \tau_1 \dots \circ \tau_{p-1}$ is \textbf{sorting} the permutation $f$ back to the identity. The existence and correctness of the \textbf{bubble sort} algorithm, which operates uniquely by performing swaps,  proves that such a sequence of swaps exists: we can take the swaps done by bubble sorting the sequence $[f(0),f(1),\dots,f(m-1)]$.
\end{proof}

\begin{restatable}[Parity of a bijection]{theorem}{thparity}
\label{th:parity}
Let $f\in \mathfrak{S}_m$. The parity of the number of swaps used in any swap-decomposition of $f$ does not depend on the decomposition. If $f= \tau_{p-1}\circ \dots \circ \tau_0$ and $f= \tau'_{p'-1}\circ \dots \circ \tau'_0$ then $p \equiv p' \; [2]$. Hence we say that the function $f$ is even if $p$ is even and odd otherwise.
\end{restatable}
The proof is in 
Appendix~\ref{app:proof_parity}.

\begin{remark}
From the points made in the proof of Lemma~\ref{lem:swde} and in Theorem~\ref{th:parity}\footnote{Also known as the ``Futurama theorem'': \url{https://www.youtube.com/watch?v=J65GNFfL94c}}, we can also interpret the parity of a bijection to be the parity of the number of swaps needed to sort it back to the identity.
\end{remark}

\begin{example}
By $f=(1,0,3,2)\in \mathfrak{S}_4$, we mean $f(0) = 1, f(1) = 0, f(2) = 3, f(3) = 2$. The bijection $f=(1,0,3,2)$ is \textbf{even} as we can sort it in 2 swaps by swapping $1$ and $0$ then $3$ and $2$. The bijection $f=(0,2,1)\in \mathfrak{S}_3$ is \textbf{odd} as we can sort it in 1 swap by swapping $2$ and~$1$. 
\end{example}

\begin{corollary}[Multiplication table]
\label{cor:table}
When looking at the parity of bijections, the following multiplication table holds: $\textbf{odd}\circ\textbf{odd} = \textbf{even}$,
$\textbf{odd}\circ\textbf{even} = \textbf{odd}$,
$\textbf{even}\circ\textbf{odd} = \textbf{odd}$,
$\textbf{even}\circ\textbf{even} = \textbf{even}$.
\end{corollary}
\begin{proof}
We give the proof for $\textbf{even}\circ\textbf{even} = \textbf{even}$, other cases are similar. Take $f,g\in \mathfrak{S}_m$ to be two even bijections. Decompose $f$ and $g$ into swaps: $f=\tau_{p-1}\circ \dots \circ \tau_0$, $g=\tau'_{q-1}\circ \dots \circ \tau'_0$. Because $f,g$ are even, we know that $p$ and $q$ are even. Note that $f\circ g = \tau_{p-1}\circ \dots \circ \tau_0 \circ \tau'_{q-1}\circ \dots \circ \tau'_0$. Hence there exists a swap-decomposition of $f\circ g$ using an even number, $p+q$, of swaps.
By Theorem~\ref{th:parity}, we conclude that $f\circ g$ is even.
\end{proof}

In this paper, we are only concerned by a very specific kind of bijections: $k$-cycles. Indeed, Lemma~\ref{lem:kcycles} will show that the circuit function of a bijective $2^n$-counter is a $2^n$-cycle. The parity of a $k$-cycle is easy to compute: it is equal to the parity of $k-1$ (Theorem~\ref{th:kcycles}).

\begin{restatable}[$k$-cycle]{definition}{defkcycles}
For $k\geq1$, a $k$-cycle $\rho \in \mathfrak{S}_m$ is a bijection such that there exists distinct $x_0,x_1,\dots,x_{k-1} \in \{0,\dots,m-1\}$ such that $
\rho(x_0) = x_1, \rho(x_1) = x_2, \dots, \rho(x_{k-1}) = x_0$ and
$\forall x \not \in \{x_0,\dots,x_{k-1}\},\, \rho(x) = x$.
\end{restatable}

\begin{restatable}[Parity of $k$-cycle]{theorem}{thkcycles}
\label{th:kcycles}
A $k$-cycle $\rho$ has the parity of the number $k-1$. It means that $\rho$ is odd iff $k-1$ is odd and $\rho$ is even iff $k-1$ is even. 

\end{restatable}
The proof is in 
Appendix~\ref{app:proof_parity}.

\begin{lemma}
\label{lem:kcycles}
The circuit function of a bijective $2^n$-counter is a $2^n$-cycle.
\end{lemma}
\begin{proof}
The mapping produced by the circuit function of a bijective $2^n$-counter is illustrated in Figure~\ref{fig:all_cycles}(a), it matches the definition of a $2^n$-cycle and generalises to any $n\in\Nset^+$.
\end{proof}

\begin{corollary}[Bijective $2^n$-counters have odd circuit functions]
\label{cor:kcycles}
The circuit function of a $2^n$-counter is an odd bijection.
\end{corollary}
\begin{proof}
We know that the circuit function of a $2^n$-counter is a bijection (Lemma~\ref{lem:bij}). We know that it is a $2^n$-cycle (Lemma~\ref{lem:kcycles}). Hence, a $2^n$-counter has the parity of $2^n-1$ which is odd (Theorem~\ref{th:kcycles}).
\end{proof}

%% file: figure-potato.tex
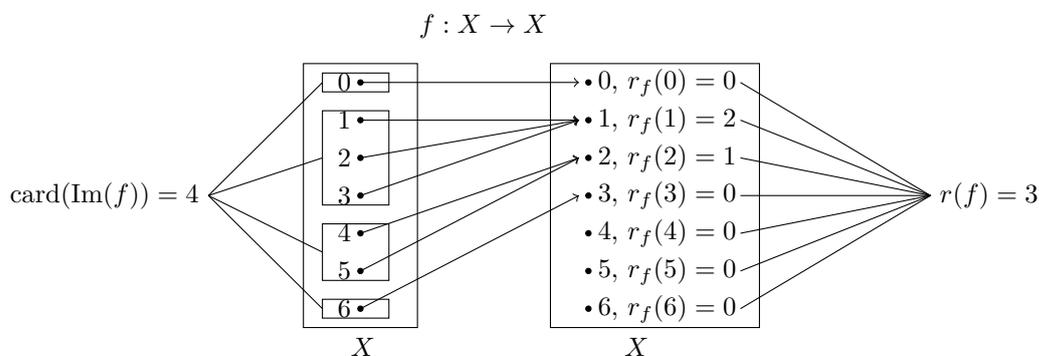
\begin{figure}[h!]
  \begin{tikzpicture}[scale=0.5]
  \begin{scope}[shift={(1,0)}]
\node [left] at (0,0) {$0$};
\filldraw (0,0) circle (2pt);
\node [left] at (0,-1) {$1$};
\filldraw (0,-1) circle (2pt);
\node [left] at (0,-2) {$2$};
\filldraw (0,-2) circle (2pt);
\node [left] at (0,-3) {$3$};
\filldraw (0,-3) circle (2pt);
\node [left] at (0,-4) {$4$};
\filldraw (0,-4) circle (2pt);
\node [left] at (0,-5) {$5$};
\filldraw (0,-5) circle (2pt);
\node [left] at (0,-6) {$6$};
\filldraw (0,-6) circle (2pt);
\draw (-1.5,0.5) rectangle (1.5,-6.5);
\draw (-1,0.25) rectangle (0.75,-0.25);

\draw (-1,-0.75) rectangle (0.75,-3.25);

\draw (-1,-3.75) rectangle (0.75,-5.25);

\draw (-1,-5.75) rectangle (0.75,-6.25);

\begin{scope}[shift={(6,0)}]
\node [right] at (0,0) {$0,\, r_f(0) = 0$};
\filldraw (0,0) circle (2pt);
\node [right] at (0,-1) {$1,\, r_f(1) = 2$};
\filldraw (0,-1) circle (2pt);
\node [right] at (0,-2) {$2,\, r_f(2) = 1$};
\filldraw (0,-2) circle (2pt);
\node [right] at (0,-3) {$3,\, r_f(3) = 0$};
\filldraw (0,-3) circle (2pt);
\node [right] at (0,-4) {$4,\, r_f(4) = 0$};
\filldraw (0,-4) circle (2pt);
\node [right] at (0,-5) {$5,\, r_f(5) = 0$};
\filldraw (0,-5) circle (2pt);
\node [right] at (0,-6) {$6,\, r_f(6) = 0$};
\filldraw (0,-6) circle (2pt);
\draw (-1,0.5) rectangle (4.5,-6.5);
\node [right] at (0.7,-7) {$X$};
\end{scope}

\draw [->] (0,0) -- (5.75,0);
\draw [->] (0,-1) -- (5.75,-1);
\draw [->] (0,-2) -- (5.75,-1);
\draw [->] (0,-3) -- (5.75,-1);
\draw [->] (0,-4) -- (5.75,-2);
\draw [->] (0,-5) -- (5.75,-2);
\draw [->] (0,-6) -- (5.75,-3);

\draw  (-1,0) -- (-4,-3);
\draw  (-1,-2) -- (-4,-3);
\draw  (-1,-4.5) -- (-4,-3);
\draw  (-1,-6) -- (-4,-3);

\node [right] at (-0.5,-7) {$X$};
\node [right] at (1.3,1.5) {$f:X \to X$};
\node [left] at (-4,-3) {$\text{card}(\text{Im}(f)) = 4$};

\draw  (10,0) -- (15,-3);
\draw  (10,-1) -- (15,-3);
\draw  (10,-2) -- (15,-3);
\draw  (10,-3) -- (15,-3);
\draw  (10,-4) -- (15,-3);
\draw  (10,-5) -- (15,-3);
\draw  (10,-6) -- (15,-3);
\node [right] at (15,-3) {$r(f) = 3$};

\end{scope}
\end{tikzpicture}
\caption{Illustration of the ramification degree for $f:X\to X$, $X=\{0,1,\dots,6\}$ and $m=7$. Also illustrates Lemma~\ref{lem:rr}: $r(f) +\text{card}(\text{Im}(f)) = \text{card}(X) = m = 7$ }\label{fig:potato}
\end{figure}

%% file: cpibc-no-odd.tex

\section{Local railway circuits do not count to $2^n$}
\label{sec:main}
Here we use the results of Section~\ref{sec:thf}, to show our main result, Theorem~\ref{th:main}, by giving two results on atomic components. The first is that atomic components are not  odd bijections. This result is known in the context of reversible circuits \cite{DBLP:conf/icalp/Toffoli80, DBLP:journals/corr/Xu15e, DBLP:conf/rc/BoykettKS16}, we give a proof that fits our framework of railway circuits. The second is that atomic components are not quasi-bijections.

\begin{lemma}[Locality restricts atomic components]
\label{lem:atom_restricted}
Let $f_{\mathcal{C}} = f_{k-1} \circ f_{k-2} \circ \dots \circ f_0$ be the decomposition into atomic components of the circuit function of a local $n$-wire railway circuit. Then we have the following:
\begin{enumerate}

\item No $f_i$ can be an odd bijection
\item No $f_i$ can be a quasi-bijection, i.e., $r(f_i) \neq 1$

\end{enumerate}
\end{lemma}

\begin{proof}

In the following, when we refer to the truth-table of $f_i$ we mean the Boolean $(n,2^n)$ matrix where the $p^{\text{th}}$ column corresponds to the bits of the $p^{\text{th}}$ element in the $2^n$-long sequence $[f_i(0,0,\cdots,0), f_i(0,0,\cdots,1), \dots, f_i(1,1,\cdots,1)]$.
\begin{enumerate}
\item Because $f_i$ is local either its first bit or its last bit has to have the properties outlined in Lemma~\ref{lem:pass}, Equations~\eqref{eq:two} and~\eqref{eq:three}. Two cases:
\begin{itemize}\item Let suppose that the first bit of $f_i$ has the properties: it is pass-through ($y_0 = x_0$) and it does not affect any other output bits than $y_0$. Then, $M$, the truth-table of $f_i$ has a very remarkable structure, the first line is composed of $2^{n-1}$ zeros followed by $2^{n-1}$ ones. Furthermore, the following $(n-1,2^{n-1})$ sub-matrices of $M$, $M_1$ and $M_2$ are equal. The sub-matrix $M_1$ is defined by excluding the first row of $M$ and taking the first $2^{n-1}$ columns while $M_2$ also excludes the first line of $M$ but takes column between $2^{n-1}$ and $2^n$. Indeed, since $x_0$ has no influence on $y_1,\dots,y_{n-1}$ we have $M_1 = M_2$. That means that we can sort $M$ in an even number of steps by using twice the sequence of swaps needed to sort $M_1 = M_2$: we first sort the first half of $M$ then transpose the swaps we used to the second half. Hence, since we can use an even number of swaps to sort $f_i$, by Theorem~\ref{th:parity}, $f_i$ is even.
\item Let suppose that the last bit of $f_i$ has the properties: it is pass-through ($y_{n-1} = x_{n-1}$) and it does not affect any other output bits than $y_{n-1}$. Again, $M$, the truth-table of $f_i$ has a very remarkable structure: an even column $p$ is such that column $p+1$ share the same first $n-1$ bits and the last bit of column $p$ is $0$ while the last bit of column $p+1$ is $1$. Column $p$ and column $p+1$ are next to each other in lexicographic order. It means that we sort the columns of $M$ by swapping blocks of two columns at each step. Since swapping two blocks of two columns can be implemented by using $2$ swaps, with this technique, we will use a multiple of $2$ swaps to sort the table. By Theorem~\ref{th:parity}, it implies that $f_i$ is even.
\end{itemize}
\item The truth-table of a quasi-bijection has the following properties: only two columns appear twice with $n$-bit vector $x_0 \in \{0,1\}^n$ and exactly one $n$-bit vector $x_1 \in \{0,1\}^n$ appears nowhere in the table. The Hamilton distance of $x_0$ and $x_1$ is at least one. Let suppose that $x_0$ and $x_1$ disagree in their $j^\text{th}$ bit, $\pi_j(x_0) \neq \pi_j(x_1)$. W.l.o.g we can take $\pi_j(x_0) = 1$. Now, because the vector $x_0$ is the only vector to appear twice in the truth table, it means that on the $j^\text{th}$ line of $M$ we see $2^{n-1}+1$ ones versus $2^{n-1}-1$ zeros, so we see an odd number of ones and an odd number of zeros. That contradict the fact that, because $f_i$ is local, $y_j$ does not depend on at least one input and hence, the number of zeros and ones on the $j^\text{th}$ line of $M$ is at least a multiple of $2$.\qedhere
\end{enumerate}
\end{proof}
We now have all the elements to prove our main result:
\thmain*
\begin{proof}
Consider the circuit function of a local $n$-wire railway circuit $\mathcal{C}$ which implements a $2^n$-counter and its decomposition into atomic components, $f_\mathcal{C} = f_{k-1} \circ f_{k-2} \circ \dots \circ f_0$.
We know that $f$ must be either a bijection or a quasi bijection (Lemma~\ref{lem:bij}), giving two cases: 
\begin{enumerate}
\item If $f$ is a bijection, by Corollary~\ref{cor:ram}, each atomic component must be a bijection too. Furthermore, by Corollary~\ref{cor:kcycles}, $f$ must be an odd bijection. But, with Lemma~\ref{lem:atom_restricted}, we know that each atomic component can only be an even bijection, and by Corollary~\ref{cor:table}, composing even bijections only leads to even bijections. Hence, $f$ cannot be odd and there are no bijective $2^n$-counter.
\item If $f$ is a quasi-bijection, by Corollary~\ref{cor:ram}, each atomic component must be either a bijection or a quasi-bijection. Furthermore we need at least one atomic component to be a quasi-bijection since composing bijections only leads to bijections. However, Lemma~\ref{lem:atom_restricted} shows that no atomic component can be a quasi-bijection. Hence, $f$ cannot be a quasi-bijection and there are no quasi-bijective $2^n$-counter.\qedhere
\end{enumerate} 
\end{proof}

%% file: self_assembly.tex

\section{Self-assembled counters}\label{sec:sa}

Here we apply the  Boolean circuit framework already established in previous sections to show limitations of self-assembled counters that work in base 2.
We give a short description of the abstract Tile Assembly Model (aTAM)~\cite{Winf98,rothemund2000program}, more details can be found elsewhere~\cite{PatitzSurveyJournal,DotCACM}. 

\subsection{Self-assembly definitions}\label{sec:sa defs}
Let $\Nset = \{0,1,\ldots\}$, $\Nset^+ = \{1,2,\ldots\}$, and $\Zset,\Rset$ be the integers and reals.

In the aTAM, one considers a set of square tile types $T$ where each square side has an associated {\em glue type}, a pair 
$(s,u)$  where $g$ is a (typically binary) string and 
$u \in \mathbb{N}$ is a glue strength. 
A tile $(t,z) \in T\times\Zset^2$ is a positioning of a tile type on the integer lattice. 
A glue is a pair $(g,z) \in G\times \Hset$ where 
$\Hset$ is the set of half-integer points $\Hset = \left\{ z \pm h \mid z \in \mathbb{Z}^2 , h \in  \{ (0,0.5) , (0.5,0) \} \right\}$ 
and 
$G$ is the set of all glue types of $T$. 
An assembly is a partial function $\alpha : \mathbb{Z}^2 \rightarrow T$, whose domain is a connected set. 
For $X\subset\Zset^2$ we let $\alpha |_X$ denote the restriction of $\alpha$ to domain $X$, i.e.\ $\alpha |_X : X\rightarrow T$ and for all $z\in X$, $\alpha(z)=\alpha |_X(z)$.

Let $\mathcal{T} = (T,\sigma,\tau)$ be an aTAM system where
$T$ is a set of tile types, $\tau \in \mathbb{N}^+$ is  the temperature and $\sigma$ is an assembly called the seed assembly. 
The process of self-assembly proceeds as follows.
A tile $(t,z)$ {\em sticks} to an assembly $\alpha$ if $z\in\Zset^2$ is adjacent in $\Zset^2$ to, but not on, a tile position of $\alpha$ and the glues of $t$ that touch\footnote{A pair of glues touch if they share the same half-integer position in $\Hset$.} glues of $\alpha$ of the same glue type have the sum of their strengths being at least $\tau$. 
A {\em tile placement} is a tuple $(t,z,\textrm{In}) \in T \times \Zset^2 \times \{N,E,S,W\}^{\leq 4}$, where 
$\textrm{In}$ denotes the $k \in \{ 1,2,3,4 \}$ tile sides which stick with matching glues and are called 
{\em input sides}; the remaining $4-k$ sides are called {\em output sides}. 
For example, a tile of type $t$ that binds at position $z$ using its north and west side would be denoted $(t,z,(N,W))$.
After the tile placement $(t,z, \mathrm{In})$ to assembly $\alpha$, 
the resulting new assembly $\alpha'= \alpha\cup \{ z\rightarrow t \}$ is said to contain the tile $(t,z)$, and we write $\alpha \rightarrow_\calT \alpha'$. 
An {\em assembly sequence} is a sequence of assemblies 
 $\vec \alpha = \alpha_0, \alpha_1 \ldots, \alpha_k$
 where for all $i$,  $\alpha_i \rightarrow_\calT \alpha_{i+1}$, in other words: 
each assembly is equal to the previous assembly plus one newly-stuck tile. 
A terminal assembly of $\calT$ is an assembly to which no tiles stick. 
$\calT$ is said to be {\em directed} if it has exactly one terminal assembly and undirected otherwise.

\subsection{Computing Boolean functions by self-assembly}
The following definitions are for representing Boolean functions as assembly systems.

\begin{definition}[bit-encoding glues]
\label{def:glues to bits}
A {\em bit-encoding glue type} is a pair  $g=(s,b) \in \Sigma^\ast  \times \{0,1, \epsilon\}$, where $s$ is a string over the finite alphabet $\Sigma$. 
If $b\neq\epsilon$, $g$ is said to encode bit~$b$, otherwise if $b=\epsilon$, we say~$g$ does not encode a bit. 
A bit encoding glue is a pair $(g,z)$ where $g$ is a bit-encoding glue type, and $z\in \Hset$ is a position.  
\end{definition}

In a tile placement, if a bit encoding glue is on an input side of the tile placement we call that an input bit to the tile  placement, if it is on an output side it is called an output bit.

\begin{definition}[Cleanly mapping tile placements to a railway circuit gate]\label{def:cleanly}
Let $A$ be a set of assembly sequences  
that use tiles with bit-encoding glues, 
let $z \in \Zset^2$ be a position, 
and $P_z$ be union of the tile placements from position $z$ over all $\vec\alpha \in A$. 
$P_z$ is said to {\em cleanly map to a railway circuit gate} if
\begin{enumerate}[(a)]
\item each $\vec\alpha$ has a tile placement at position $z$; and 
\item there is a $k\in \{0,1,2 \}$ such that all placements in $P_z$ map $k$ input bits to $k$ output bits; and
\item if an $\epsilon$-glue $g$ (non-0/1 encoding glue) appears at direction $d \in \{ N,W,S,E \}$  for some $p\in P_z$ then all $p'\in P_z$ have glue $g$ at direction $d$.
\end{enumerate}
\end{definition}

\begin{remark}
The previous definition is crafted to allow glues to encode bits in a way that can be mapped to railway circuit  gates, but also to prevent tiles from exploiting non-bit-encoding glues to ``cheat'' by working in a base  higher than 2.
\end{remark}

\begin{definition}[glue curve]
A {\em glue curve} $c : (0,1)\rightarrow\Rset^2$ is an infinite-length simple curve\footnote{A simple curve is a mapping from the open unit interval $(0,1)$ to $\Rset^2$ that is injective (non-self-intersecting).}
that starts as of a vertical ray from the south, then has a finite number of unit-length straight-line segments 
that each trace along a tile side (and thus  each touch a single point in $\Hset$), 
and ends with a vertical ray to the north.
\end{definition}
By a generalisation of the Jordan curve theorem to infinite-length simple polygonal curves, 
a glue curve $c$ cuts the $\Rset^2$ plane in two~\cite{paths2020} (Theorem~B.3). 
We let $\mathrm{LHS}(c) \subsetneq \mathbb{Z}^2$ denote the points of  $\mathbb{Z}^2$ that are on the left-hand side\footnote{Intuitively, the left-hand side of a glue curve (or any an infinite simple curve) is the set of points from $\Rset^2$ that are on one's left-hand side as one walks along the curve, excluding the curve itself. See~\cite{paths2020} for a more formal definition appropriate to glue curves.} of $c$, and 
$\mathrm{RHS}(c) \subsetneq \mathbb{Z}^2$ denote the points of  $\mathbb{Z}^2$ that are on the right-hand side of $c$.
For a vector $\vec v$ in $\Rset^2$ we define 
$c + \vec v$ to be the curve with the same domain as $c$ (the interval $(0,1)$) but with image $\cup_{x\in (0,1)}  c(x)+\vec v$ 
i.e. the translation of the image of $c$ by $\vec v$.

For an assembly $\alpha$ and glue curve $c$ 
let $\gac{\alpha}{c} = g_0, g_1,\ldots$ denote the sequence of all glues of tiles of $\alpha$ that are positioned on $c$, written in $c$-order. 
Let $\calB : \{ g \mid g \textrm{ is a sequence of bit-encoding glues} \} \times \{ c \mid c \textrm{ is a glue curve} \} \rightarrow \{ 0,1\}^\ast$, 
where for the glue sequence  
$g = (s_0,b_0), (s_0,b_1),\ldots$ we define 
$\calB( g ) = b_0 b_1\dots \in \{0,1 \}^\ast$ to be the bit-string encoded by the glues $g$ (here if some 
$b=\epsilon$ it is interpreted as the empty string). 
Hence, $\calB( \gac{\alpha}{c} )$ is the sequence of bits encoded along glue curve $c$ by assembly $\alpha$.

\begin{definition}[Layer-computing an $n$-bit function]\label{def:assemble a layer function}
Let $n\in\Nset$.
The tile set $T$ is said to {\em layer-compute} the $n$-bit function $f : \{0,1\}^n\rightarrow\{0,1\}^n$  
if there exists  
a temperature $\tau\in\Nset$,
a glue curve $c$, 
and a vector $\vec v \in \Rset^2$ with positive $\mathrm{x}$-component such that, 
for all $i \in \{1,2,\ldots \}$, $\img(c) \cap \img(c+i\cdot v) = \emptyset$, 
and for all $x \in \{0,1\}^n$
there is an assembly $\sigma_x$ positioned on the left-hand side of $c$, 
such that the tile assembly system $\calT_x = (T,\sigma_x,\tau)$  is directed and 
\begin{enumerate}[(a)]
\item\label{def:condition:glue encoding} 
$\calB(\gac{\sigma_x}{c})=x$ and  
$\calB(\gac{\alpha_x}{(c+\vec v)})=f(x)$ where $\alpha_x$ is the terminal assembly of $\calT_x$; and  

\item\label{def:condition:niceseq} 
$\calT_x$ has at least one assembly sequence $\vec{\alpha}_x$ that assembles all of $\alpha |_{\mathrm{LHS}(c + \vec v)}$ without placing any tile of $\alpha |_{\mathrm{RHS}(c + \vec v)}$; and 

\item\label{def:condition:tile placements map to gate} 
for all positions $z\in ( \mathrm{RHS}(c) \cap \mathrm{LHS}(c + \vec v) )\subsetneq \Zset^2$, 
and for all $x$
the set $P_z$ of tile placements at position $z$ in $\vec\alpha_x$ (from~(\ref{def:condition:niceseq}))
map cleanly to a gate (Definition~\ref{def:cleanly}).
\end{enumerate}
\end{definition}

\begin{remark}
Definition~\ref{def:assemble an iterated layer function} and Lemma~\ref{lem:self-assembly of a layer fn} below are not needed to prove our main result, but are included to show that any system satisfying Definition~\ref{def:assemble a layer function} can be ``iterated'' to compute layer-by-layer, similar to the sense in which tile-based counters in the literature compute layer-by-layer.
\end{remark}

\begin{definition}[Computing an iterated $n$-bit function]\label{def:assemble an iterated layer function}
The tile set $T$ is said to layer-compute $f^k$, 
i.e., $k$ iterations of the Boolean function $f : \{0,1\}^n\rightarrow\{0,1\}^n$,  
 if 
\begin{enumerate}[(a)]
\item $T$ layer-computes $f$ via Definition~\ref{def:assemble a layer function}, and 
\item $\calT_x$'s terminal assembly $\alpha_x$, for each $i\in \{0,1,\ldots,k-1\}$, has
$\calB(\gac{ \alpha_x}{(c + i \cdot \vec v)}) = f^i(x)$.
\end{enumerate}
\end{definition}
\begin{lemma}\label{lem:self-assembly of a layer fn}
Let $T$ be a tile set that  layer-computes the $n$-bit function $f$ 
according to Definition~\ref{def:assemble a layer function}. 
Then for any $k\in\Nset$, $T$ also computes 
$f^k$ according to Definition~\ref{def:assemble an iterated layer function}.
\end{lemma}
\begin{proof} 
We give a simple proof by induction on iteration $f^i$ for $i \in \{0,1,\ldots,k-1 \} $. 
The intuition is to let $y' = f^i(x) \in \{0,1\}^n$, then use Definition~\ref{def:assemble a layer function} 
to compute $f^1(y')$, and then translate the resulting assembly to a later layer to compute $f^{i+1}(y') = f( f^{i}(y') )$.

For the base case, let $i=0$. By Definition~\ref{def:assemble a layer function}, the seed $\sigma_x$ encodes $x\in\{ 0,1\}^n$ along $c$ as $\calB(\gac{\sigma_x}{c})=x$ and $\sigma_x$ places no tiles in $\mathrm{RHS}(c) \subsetneq \Zset^2$. 
Hence $T$ computes $f^0(x)=x$ according to Definition~\ref{def:assemble an iterated layer function}.

For the inductive case, let $i > 0$.
Let $\vec{\alpha} = \alpha_0,\alpha_1,\ldots,\alpha_m$ be an assembly sequence of $\calT_x$ such that, inductively, we assume that on the last assembly $\alpha_m$ of $\vec\alpha$ the cut $c + i \cdot \vec v$ encodes the bit sequence 
$f^i(x) = \calB(\gac{\alpha_x}{(c+\vec i \cdot v)}) \in \{0, 1\}^n$, 
and $\alpha_m$ has no tiles in $\mathrm{RHS}(c + i \cdot \vec v)$, and 
no more tiles can stick in $\mathrm{LHS}(c + i \cdot \vec v)$.
Next, consider the seed assembly $\sigma_{y'}$ that encodes $y' = f^i(x)$, and let 
$\vec{\alpha} = \sigma_{y'}, \alpha_1, \alpha_2, \ldots $ 
be an assembly sequence that satisfies   
Definition~\ref{def:assemble a layer function}(\ref{def:condition:niceseq}) and (\ref{def:condition:tile placements map to gate}).
For each $k \in \{ 0, 1,\ldots, |\vec{\alpha} |-1 \}$ 
let $(t,z)_k$ be the tile that sticks (is attached) to assembly 
$\alpha_k$ to give the next assembly $\alpha_{k+1}$ in~$\vec\alpha$. 
We make a new assembly sequence $\vec{\alpha}'$
that starts with the assembly $\alpha_m$ (defined above), 
then stick the `translated tile' 
$(t,z + i\cdot \vec v)_0$, 
and then sticks $(t,z + i\cdot \vec v)_1$, and so on.
  In other words define
  $\vec{\alpha}' = \alpha_{m,0} \rightarrow_\calT \alpha_{m,1} \rightarrow_\calT \cdots $ by beginning with 
  $ \alpha_{m} = \alpha_{m,0}$ and in turn attaching the tiles 
$(t,z + i\cdot \vec v)_0, (t,z + i\cdot \vec v)_1,\ldots $
 in order. 
 the assembly sequence 
$\vec{\alpha}'$ 
eventually contains an assembly that encodes $f(y') = f^{i+1}(x)$ along the curve $c + (i+1)\cdot \vec v$ as $f^{i+1}(x) = \calB(\gac{\alpha_x}{(c+ (i+1)\cdot\vec v)})$. 
This completes the induction.
\end{proof}

\newpage
\subsection{Tile sets computing layer-by-layer are simulated by railway circuits} 

\begin{lemma}[Railway circuits simulate layer-computing tile sets]\label{lem:tiles to local circuit}
Let $T$ be a tile set that layer-computes, via Definition~\ref{def:assemble a layer function}, the Boolean function $f : \{0,1\}^n\rightarrow\{0,1\}^n$ for $n \geq 3$.
Then there is a local railway circuit that computes $f$.
\end{lemma}
\begin{proof}
By Definition~\ref{def:assemble a layer function}, 
let $\tau$ be the temperature, 
$c$ the glue curve, 
$\vec{v}$ the vector,
and for each input $x\in \{0,1\}^n$ 
let $\sigma_x$ be the seed assembly encoding $x$ and let $\calT_x=(T,\sigma_x,\tau)$. 
We will construct a $n$-wire local railway circuit from $T$. 

For each $x$, let $\vec\alpha_x$ be the assembly sequence that satisfies 
Definition~\ref{def:assemble a layer function}(\ref{def:condition:niceseq})
and (\ref{def:condition:tile placements map to gate}).
By directedness, the choice of $\vec{\alpha}_x$ over other assembly sequences is inconsequential since all assembly sequences for $x$ compute the same output $f(x)$ for the layer.
For any $z\in ( \mathrm{RHS}(c) \cap \mathrm{LHS}(c + \vec v) )\subsetneq \Zset$
and for all $x \in \{0,1\}^n$, 
let $P_z$ denote the set of tile placements that appear at position $z$ in the set of assembly sequences $\cup_x (\vec\alpha_x)$.
By Definition~\ref{def:assemble a layer function}(\ref{def:condition:tile placements map to gate}), $P_z$ maps cleanly to a gate, and we let $g_z$ denote that gate. 
Via Definition~\ref{def:cleanly}, there are $k_z \in \{0,1,2\}$ inputs and $k_z$ outputs to $g_z$ (i.e. fanin and fanout are equal to $k_z \leq  2$ for~$g_z$).
In the railway circuit for each gate  $g_z$ we define a section $s_z$; there are $n$ inputs and $n$ outputs to the section, 
$k_z$ of those $n$ are fed through gate $g_z$ and the remaining  $n - k_z$ are pass-through.  
In other words, the $n$-bit function computed by the section is defined by the extension $f_{g_z}\!\!\upharpoonright$ to $n$ bits of the function $f_{g_z}$ computed by gate $g_z$ on $k_z$ bits (see Section~\ref{sec:circuits}). 
The section is local since $k_z < n$. 

It remains to wire the sections together (order them) so that they compute $f$.
Choose any $x \in \{0,1\}^n$, and let $z_0, z_1, \ldots,z_m$ be the sequence of positions 
in the order defined by the canonical assembly sequence $\vec\alpha_x$ as defined earlier in the proof.
Let $x' \in \{0,1\}^n$ where $x' \neq x$ and let $\vec\alpha_x'$ be the canonical assembly sequence for $x'$. 
By Definition~\ref{def:cleanly}, all assembly sequences share the same set of tile positions from $\Zset^2$.   
We will define a new assembly sequence $\vec\beta_x'$, that has $\sigma_x'$ as its first assembly,
fills positions in the order $z_0, z_1, \ldots,z_m$ (we used for $\vec\alpha_x$)
and produces the same terminal assembly as $\vec\alpha_x'$. 
For $z$ in $z_0, z_1, \ldots,z_m$, in the order given, let $p$ be tile placement at position $z$ in $\alpha_x'$ ($p$ exists by Definition~\ref{def:cleanly}),
and we claim we can attach $p$ to the current assembly $\beta$ to get a new assembly $\beta'$. 
Since, the assembly $\alpha$ from $\vec\alpha_x$ that receives the placement at position $z$ to make $\alpha'$ (i.e. $\alpha\rightarrow_\calT \alpha'$), 
has neighbouring positions providing sufficient tile sides (glues) for the placement at $z$, 
and since we are iterating through placement positions in the same order, 
the same is true of~$\beta$.  
Hence $p$ can be placed on $\beta$ to give $\beta'$. 
Since $\vec\alpha_x$, $\vec\alpha_x'$ and $\vec\beta_x'$ all share the same set of positions, 
and since $\vec\alpha_x'$ and $\vec\beta_x'$ share the same tile placement at each shared position, 
the terminal assembly of $\vec\alpha_x'$ and $\vec\beta_x'$ are identical. 
Each position $z_0, z_1, \ldots,z_m$  defines a section $s_{z_0}, s_{z_1}, \ldots, s_{z_m}$, and we wire the sections in the order given.
Since each section is a local railway circuit from $n$ bits to $n$ bits, their composition/concatenation is too. 
By Definition of $T$, $\calB(\gac{\alpha_x'}{(c + \vec v)})= f(x) \in \{0,1\}^n$, and by construction the output of section $s_{z_m}$ is the same value $f(x)$.
 \end{proof}

\begin{theorem}\label{thm:self-assembly:no odd bijection}
Let $n\in \{5,6,\ldots \}$.
There is no tile set $T$ that layer-computes, via Definition~\ref{def:assemble a layer function}, 
an odd bijective $n$-bit function, or a quasi-bijective $n$-bit function, or a 
$2^n$ counter on $n$-bits. 
\end{theorem}
\begin{proof}
By Lemma~\ref{lem:tiles to local circuit} there is a railway circuit $\calC$ that layer-computes the same function, $f$, as  $T$.
By Theorem~\ref{lem:atom_restricted}, $\calC$ does not compute an odd bijection nor quasi-bijection, and hence (or by Theorem~\ref{th:main})
 $\calC$ does not compute a $2^n$ counter. 
\end{proof}

\subsection{Examples: application of our main self-assembly result}
We illustrate our definitions by applying them to previously studied self-assembly systems. 

\subsubsection{Example: IBC tile sets}

\begin{figure}
\includegraphics[width=\textwidth]{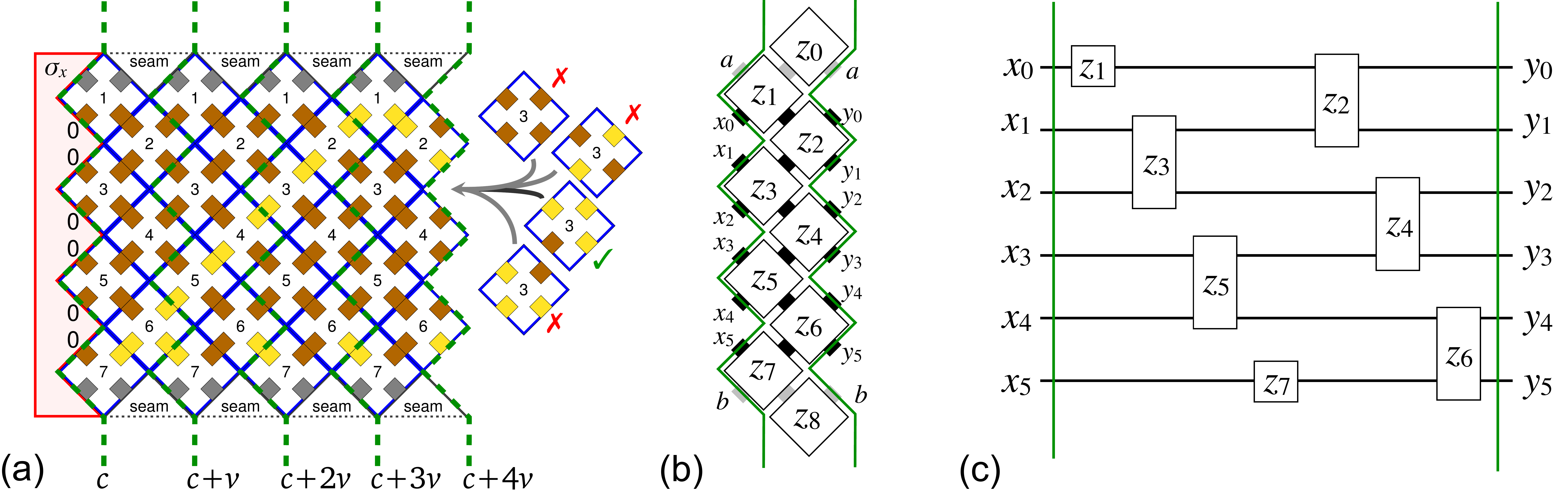}
\caption{Impossibility of a $2^n$ counter in IBCs. The lattice is rotated by $45^\circ$, relative to the standard $\Zset^2$ lattice.
Theorem~\ref{th:main:ibc}, shows that although there are 64 possible bit-strings that appear across the all cuts for all 6-bit systems, no 64-counter is possible in any one system. 
(a) First few layers of a 6-bit 63-counter for the 6-bit IBC tile model that appeared in~\cite{ibc}.
The glue curve $c$, and its translations by multiples of vector $\vec v$ are shown in dashed green.  
(b) Layer for the same system showing tile positions $z_0,z_1,\ldots,z_8 \in \Zset^2$. Glues that encode 0/1 bits are shaded black, $\epsilon$-glues (that do not encode a bit) are shaded grey, for use in Definition~\ref{def:assemble a layer function}. 
(c) Local railway circuit that simulates the IBC via Lemma~\ref{lem:tiles to local circuit}, the existence of that  
that simulation implies the impossibility of a $2^n$-counter.  
The negative result generalises to all even $n$ and any number of tile layers (depth) $\geq 1$.}\label{fig:ibc}
\end{figure}

\begin{example}[IBC tile sets]\label{eg:ibc}
For $n,\ell\in\Nset$, with $n$ even, the directed\footnote{It is possible to have randomised and non-randomised/deterministic/directed  IBCs~\cite{ibc}, here we look only at the deterministic case, meaning that given a seed and tile set, an IBC produces exactly one terminal assembly.} $n$-bit $\ell$-layer IBC model is defined in 
Section SI-A-S1 of~\cite{ibc}, and the  6-bit 1-layer IBC tile assembly model is defined graphically in Figure 1b of the same paper.
Here it is illustrated in Figure~\ref{fig:ibc}.  
The IBC model is a restriction of the aTAM, hence we use the terminology from Section~\ref{sec:sa defs}. 
A directed 6-bit IBC tile set $T$ is a set of 31 tile types with 4 tiles for each of positions $z_2,z_3,\ldots,z_6$ (mapping two input bits to two ouput bits), 
two tiles for each of positions $z_1$ and $z_7$ (mapping one input  bit to one output bit), and one seam tile for each of positions $z_0$ and $z_8$ (mapping no input bits to no output bits -- using $\epsilon$ glues in Definition~\ref{def:glues to bits}). 
The set of tile types associated to a position is unique to that position, hence in an assembly tile types that appear on row $i$ will never appear on row $j \neq i$. 
We let $P_{z_i}$ denote the set of tile placements for position $z_i$, where $i \in \{0,1,2,\ldots,8 \}$, Figure~\ref{fig:ibc}(a) illustrates the 4 tile placements for position $z_3$. 

More generally, for any even $n$, and $\ell \geq 1$, the $n$-bit, $\ell$-layer IBC model is defined in \cite{ibc}.
For the $\ell$-layer case, there are 
$2 \ell$ tiles that map 1 input bit to 1 output bit,  
$(n-1) \ell$ tiles that map 2 input bits to 2 output bits,
and $2\ell$ tiles (or merely $\ell$ tiles if using a tube topology as in \cite{ibc}) that map 0 input bits to 0 output bits.  
\end{example}

\begin{lemma}\label{lem:ibc to railway}
For any even $n \in \{ 2n' \mid n\in \Nset \} $, $\ell \in\Nset^+$, an $n$-bit IBC tile set $T$ layer-computes a Boolean function according to Definition~\ref{def:assemble a layer function}. 
\end{lemma}
\begin{proof}
IBCs (described in Example~\ref{eg:ibc}, illustrated in in Figure~\ref{fig:ibc}) satisfy Definition~\ref{def:assemble a layer function}: 
In that definition let $\tau = 2$, 
let $c$ run along the seed $\sigma_x$ as shown in Figure~\ref{fig:ibc}, 
and let $\vec v = (0, \ell \sqrt{2})$ (assuming the x-axis is horizontal).
Next, note that Definition~\ref{def:assemble a layer function}(\ref{def:condition:glue encoding}) is satisfied by $c$, $\vec v$, 
$\sigma_x$ for each $x\in \{0,1\}^n$, 
and the fact that the tile set outputs a bit sequence (that we define to be $f(x)$) along $c +\vec v$. 
Definition~\ref{def:assemble a layer function}(\ref{def:condition:niceseq}) is satisfied by the assembly sequence
that places tiles in ``half-layer order'' as follows: $z_1,z_3,\ldots,z_{n+1}$ then $z_0,z_2,\ldots,z_{n+2}$, for layer 1, and so on for all $\ell$ layers.
Definition~\ref{def:assemble a layer function}(\ref{def:condition:tile placements map to gate}):  
for $i \in \{0,1,\ldots, (n+3)\ell -1 \}$, the sets of tile placements $P_{z_i}$ are defined using (a fairly obvious) generalisation of the scheme illustrated in Example~\ref{eg:ibc} and Figure~\ref{fig:ibc}(b) for $n=6$ and $\ell=1$.
Each such placement maps cleanly to a railway circuit gate since IBC tile positions each have an associated set of $2^k$ tile types that each map $k\in\{0,1,2\}$ input bits to $k$ output bits, in particular, fanin is equal to fanout for each gate, and $< n$ hence for $n\geq 3$ the railway circuit is local. 
 \end{proof}
 Thus, via Lemma~\ref{lem:tiles to local circuit}, 
 $n$-bit IBC tile sets are simulated by $n$-wire local railway circuits (Figure~\ref{fig:ibc}(c) shows an example).
 Hence by Theorem~\ref{thm:self-assembly:no odd bijection} 
 we immediately get:\thmainibc*

\subsubsection{Example: zig-zig tile sets}\label{sec:zigzig}
\begin{figure}[t]
\includegraphics[width=\textwidth]{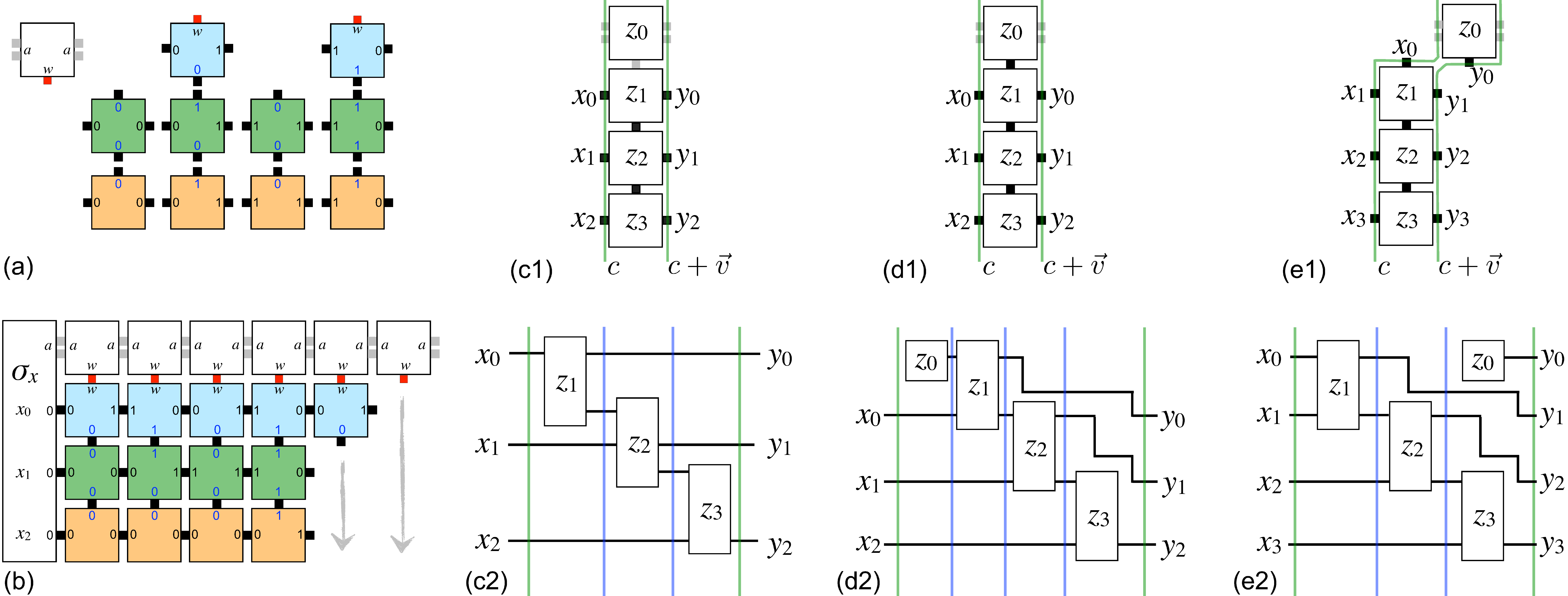}
\caption{Impossibility and possibility of  a $2^n$ zig-zig counter in the aTAM.
(a) An aTAM tile set.
Glues that encode 0/1 bits are shaded black, $\epsilon$-glues (that do not encode a bit) are shaded grey. 
Red denotes a glue type $w$ that can be either be a 0/1 encoding glue (assembles a $2^{n-1}$ counter -- non-maximal), 
or an $\epsilon$-glue (assembles a $2^{n}$ counter -- maximal). 
Each row has a unique set of tile types, indicated by tile colour. 
(b)~Example growth starting from a seed assembly $\sigma_x$ that encodes the input $x = x_0 x_1 x_2 = 000$. 
(c1)~A layer defined by the curve $c$ and its translation $c+\vec v$, both in green, and (c2) its simulation by a circuit. 
The circuit is not a railway circuit since gates $z_1$ and $z_2$ have fanout unequal to fanin; 
this can be seen by counting the number of wires that intersect each section border (3 wires on the green borders, 4 on the blue).  
The layer in (c1) and the circuit in (c2) define a maximal $2^n$-counter for $n=3$, 
and the construction generalises to give a maximal $2^n$-counter for any $n \in \Nset$).
Likewise, (d1) and (d2) define a $2^n$ counter by exploiting unequal fanin and fanout on some gates: 
in particular, in (d1) we've chosen the $w$ glue to be a 0/1-encoding glue and the resulting circuit in (d2) is not a railway circuit 
(the green section borders intersect 3 wires, the blue intersect 4 wires). 
Finally, in (e1) we define $c$ and $\vec v$ in a way that gives a $2^{n-1}$ counter 
(since setting $x_0=1$ enables counting on 3 bits, and setting $x_0=0$ does not help -- forces the other bits to merely be copied). In this case the resulting circuit in (e2) is not a valid railway circuit (neither maximality nor application of our main result).}\label{fig:zigzig}
\end{figure}

Figure~\ref{fig:zigzig} illustrates a simple ``zig-zig'' counter system, where each column of tiles increments an $n$-bit binary input, for $n=3$. 
By repeating the rows of green tiles (either by using the same title types or hardcoding rows) the system generalises to arbitrary $n \in \Nset$.   

Our main self-assembly result does not apply to zig-zig systems. 
Figure~\ref{fig:zigzig}(c1), (d1) and (e1) show a number of choices for 0/1-encoding glues, versus $\epsilon$-glues, as well as two choices for the curve $c$, and in the three cases our attempt to construct a railway circuit fails. 
The circuit (and tile types) exploit unequal gate fanin and gate fanout, hence some positions do not map cleanly to a gate hence Definition~\ref{def:assemble a layer function} does not apply. 
Furthermore, it can be seen that for any $n\in\Nset^+$ a maximal $2^n$ counter is achieved (shown for $n=3$ in Figure~\ref{fig:zigzig}).

\subsubsection{Example: zig-zag tile sets}\label{sec:zigzag}

\begin{figure}[h!]
\includegraphics[width=\textwidth]{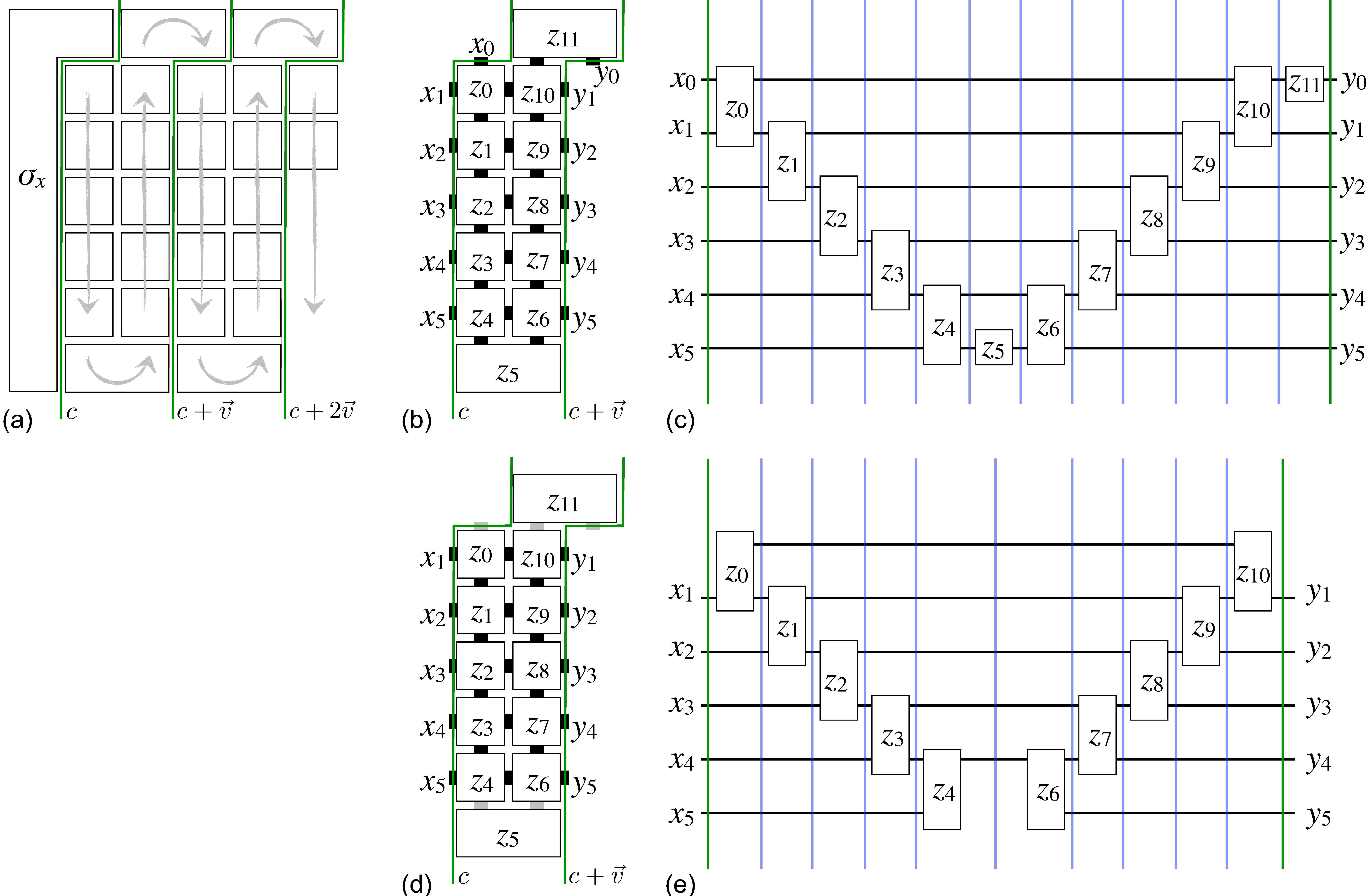}
\caption{Zig-zag tile assembly system, similar to Evans~\cite{evans2014crystals}. (a)~Schematic of a zig-zag system showing the seed $\sigma_x$ and arrows indicating tile attachment order. 
The first column of tiles (``zig'') implements a binary increment, using tile types similar to those in Figure~\ref{fig:zigzig}(a),
the second column of tiles (``zag'')  copies a columns of input bits to the right.    
The choices for the curve $c$ (in green) and vector $\vec v = (2,0)$ for Definition~\ref{def:assemble a layer function} are shown. 
(b)~A layer consists of a single zig (tiles at positions $z_0,\ldots,z_5$) followed by a zag (tiles at position $z_6,\ldots,z_{11}$,). 
Intuitively, because some glues output constant bits (e.g. $x_0=y_0=1$ we are free to choose whether 
those glues should be 0/1-encoding glues (black), and $\epsilon$-glues (grey), in (b) we have all glues be 
0/1-encoding which leads to a valid 
railway circuit in (c) that simulates the tile layer in (b).
Specifically, in (c) all gates having fanin and fanout that is equal; in other words the functions 
from one green/blue cut to the next green/blue cut are all from $n$ bits to $n$ bits.
Moreover, no gate spans all 6 wires hence the railway circuit is local.
Hence our main theorem applies that this system does not implement a $2^n$ counter on $n=6$ bits (it does however, implement a $2^{n-1}$-counter on $n=6$ bits).
(d)~If we instead assume that $x_0, y_0$ are $\epsilon$-glues, we get a maximal $2^n$ counter (but on only $n=5$ bits). 
In (e) the resulting circuit is not a railway circuit since some gates $z_0, z_4, z_6, z_{10}$ have unequal fanin and fanout.
Unequal fanin and fanout means that the function from blue/green cut to blue/green cut are not all on a fixed number $n$ bits, hence our techniques do not apply. 
}\label{fig:zigzag}
\end{figure}

Figure~\ref{fig:zigzag}(a) illustrates an aTAM schematic of a ``zig-zag'' counter system  
that was implemented experimentally in~\cite{evans2014crystals}. 
The system has alternating  increment (``zig'') and copy (``zag'') columns. 
The increment columns use similar tiles to those shown in Figure~\ref{fig:zigzig}(a). 
If we fix $n$ (e.g.\ for Figure~\ref{fig:zigzig}, let $n=6$), and vary our interpretation of the glues as either 0/1-encoding or $\epsilon$-glues, 
the counter can be seen to implement a non-maximal $2^{n-1}$ counter on $n$ bits, or a maximal $2^n$ counter on $n-1$ bits. 

Specifically, within each zig-zag layer, Figure~\ref{fig:zigzag}(b) interprets glue positions shown in black (e.g.~$x_0$ and $y_0$) as encoding a bit,
and in this case the system meets Definition~\ref{def:assemble a layer function}, and via the proof of Lemma~\ref{lem:tiles to local circuit} we get the railway circuit shown in Figure~\ref{fig:zigzag}(c). 
Hence, with that glue interpretation a $2^n$ counter is impossible (Theorem~\ref{thm:self-assembly:no odd bijection}). 
Further intuition be obtained due from the tile set design: 
in Figure~\ref{fig:zigzig}(a) the bit $x_0 = y_0$ and is always 1, and an analysis of the tile set shows that we get a $2^{n-1}$-counter. 

If we instead, use the interpretation in Figure~\ref{fig:zigzag}(d) interprets several of the glue positions (e.g.~$x_0$ and $y_0$) as not encoding a bit, and instead being $\epsilon$-glues. 
In this case our attempt to apply Definition~\ref{def:assemble a layer function} fails as some tile positions map to Boolean gates with fanin unequal to fanout; see
for example gates $z_0,z_4,z_6,z_{10}$ in Figure~\ref{fig:zigzag}(e). 
Hence our techniques do not apply. 
An analysis of the tile set shows that we get a maximal counter (but on one fewer bit that the sub-maximal counter above). 
This example shows that a system that sticks to our formalism, except for the fanin/fanout criteria, may exhibit sufficient expressive capabilities to achieve a maximal counter.

%% file: app-proof_parity.tex
\newpage
\section{Parity of a bijection}
\label{app:proof_parity}

For the sake of completeness we prove  Theorems~\ref{th:parity} and \ref{th:kcycles} from Section~\ref{sec:thf}:

\thparity*
\thkcycles*

These are known group theoretical results and the literature offers a lot of different proofs\footnote{See this thread:\\ \scriptsize{\url{https://math.stackexchange.com/questions/46403/alternative-proof-that-the-parity-of-permutation-is-well-defined}}} for them \cite{rotman2012introduction}. The following lemma is crucial:

\begin{lemma}[Sign of a bijection]
\label{lem:sign}
Let $f\in\mathfrak{S}_m$. Define $\epsilon: \mathfrak{S}_m \to \{-1,1\}$, the sign of $f$, to be:
$$ \epsilon(f) = \frac{\prod_{0 \leq j < i < m} (j-i)}{ \prod_{0 \leq j < i < m} (f(j)-f(i))}$$
We have:
\begin{enumerate}
\item $\epsilon(f) = 1$ or $\epsilon(f) = -1$
\item Let $g\in\mathfrak{S}_m$, then $\epsilon(f\circ g) = \epsilon(f)\epsilon(g)$
\item Let $\tau$ be a swap then $\epsilon(\tau) = -1$
\end{enumerate}
\end{lemma}
\begin{proof}
\begin{enumerate}
\item Because $f$ is a bijection, the sets $\{\{j,i\} \, | \, 0 \leq j < i < m\}$ and $\{\{f(j),f(i)\} \, | \, 0 \leq j < i < m\}$ are the same. However, the sets of ordered pairs $\{(j,i) \, | \, 0 \leq j < i < m\}$, $\{(f(j),f(i)) \, | \, 0 \leq j < i < m\}$ might differ when $f(j) > f(i)$, i.e when $f$ reverses the order of $(j,i)$. Hence $\epsilon(f) = 1$ if $f$ reverses the order an even number of times and $\epsilon(f) = -1$ if $f$ reverses the order an odd number of times.
\item We have $$\epsilon(f\circ g) = \frac{\prod_{0 \leq j < i < m} (j-i)}{ \prod_{0 \leq j < i < m} (f\circ g(j)-f\circ g(i))} = \frac{\prod_{0 \leq j < i < m} (j-i)}{ \prod_{0 \leq j < i < m} ( g(j)-g(i))} \frac{\prod_{0 \leq j < i < m} ( g(j)-g(i))}{ \prod_{0 \leq j < i < m} (f\circ g(j)-f\circ g(i))} $$
But because $g$ is a bijection, we have $$\frac{\prod_{0 \leq j < i < m} ( g(j)-g(i))}{ \prod_{0 \leq j < i < m} (f\circ g(j)-f\circ g(i))} = \frac{\prod_{0 \leq j < i < m} (j-i)}{ \prod_{0 \leq j < i < m} (f(j)-f(i))}$$
So we have: $$\frac{\prod_{0 \leq j < i < m} (j-i)}{ \prod_{0 \leq j < i < m} (f\circ g(j)-f\circ g(i))} = \frac{\prod_{0 \leq j < i < m} (j-i)}{ \prod_{0 \leq j < i < m} ( g(j)-g(i))} \frac{\prod_{0 \leq j < i < m} (j-i)}{ \prod_{0 \leq j < i < m} (f(j)-f(i))}$$
In other words: $\epsilon(f\circ g) = \epsilon(f)\circ\epsilon(g)$.

\item Let's consider a swap $\tau$ which swaps $i_0$ and $i_1$ with $i_0 < i_1$. We have $\frac{i_0 - i_1}{i_1 - i_0} = -1$. We juste need to focus on un-ordered pairs that features $i_0$ or $i_1$ since $\tau$ leaves all other elements unchanged. Now, three cases: 
\begin{enumerate}
\item Let's consider $i_2$ such that $i_2 < i_0 < i_1$. We have $\frac{i_2 - i_0}{i_2 - i_1} = 1$. We also have $\frac{i_2 - i_1}{i_2 - i_0} = 1$.
\item Let's consider $i_2$ such that $i_0 < i_2 < i_1$. We have $\frac{i_0 - i_2}{i_1 - i_2} = -1$ but we also have $\frac{i_2 - i_1}{i_2 - i_0} = -1$ which compensates.
\item Let's consider $i_2$ such that $i_0 < i_1 < i_2$. We have $\frac{i_0 - i_2}{i_1 - i_2} = 1$. We also have $\frac{i_1 - i_2}{i_0 - i_2} = 1$.
\end{enumerate}

In all those cases, the sign is not affected: either it is compensated either it is positive. The only part of $\epsilon(f)$ with a negative sign which is not compensated is $\frac{i_0 - i_1}{i_1 - i_0} = -1$. Hence $\epsilon(\tau) = -1$. We gave the proof for $i_0 < i_1$, the argument is symmetric and can be adapted to the case $i_1 < i_0$.\qedhere
\end{enumerate}
\end{proof}
\begin{remark}
Without saying it we proved that $\epsilon: \mathfrak{S}_m \to \{-1,1\}$ is a \textbf{group morphism} between groups $(\mathfrak{S}_m,\circ)$ and $(\{-1,1\},\times)$. In fact, it is the only non-trivial one (i.e. not the identity), see \cite{rotman2012introduction}.
\end{remark}

Theorem~\ref{th:parity} becomes a piece of cake:

\thparity*
\begin{proof}
Let $f\in\mathfrak{S}_m$ and $f= \tau_{p-1}\circ \dots \circ \tau_0$ and $f=\tau'_{p'-1}\circ \dots \circ \tau'_0$. By Lemma~\ref{lem:sign} we have:
\begin{align*}
\epsilon(f) &= \epsilon(\tau_{p-1}\circ \dots \circ \tau_0) = \epsilon(\tau_{p-1}) \times \dots \times \epsilon(\tau_0) = (-1)^p\\
\epsilon(f) &= \epsilon(\tau'_{p'-1}\circ \dots \circ \tau'_0) = \epsilon(\tau'_{p'-1}) \times \dots \times \epsilon(\tau'_0) = (-1)^{p'}
\end{align*}
Hence we must have $p \equiv p' \; [2]$. 
\end{proof}
\begin{remark}
The bijection $f$ is even if $\epsilon(f) = 1$ and odd if $\epsilon(f) = -1$.
\end{remark}
Now we compute the parity of a $k$-cycle:
\defkcycles*
\thkcycles*
\begin{proof}
Let $\rho$ be a $k$-cycle acting on $x_0,x_1,\dots,x_{k-1}$. One can decompose $\rho$ is $k-1$ transposition: $\rho = \tau_{x_0,x_1} \circ \tau_{x_1,x_2} \circ \dots \circ \tau_{x_{k-2},x_{k-1}}$. Where $\tau_{x_i,x_j}$ swaps $i$ and $j$. Hence $\epsilon(\rho) = (-1)^{k-1}$ and $\rho$ is even iff $k-1$ is even.
\end{proof}